\documentclass{article}
\usepackage{a4wide}
\usepackage{amsmath,amsthm,amssymb}
\usepackage[latin1]{inputenc}
\usepackage{mathptmx,authblk}
\usepackage[textwidth=3cm]{todonotes}
\usepackage{algorithm,algorithmicx,algpseudocode}
\usepackage{graphicx}

\newtheorem{thm}{Theorem}
\newtheorem{lem}{Lemma}
\newtheorem{conj}{Conjecture}
\newtheorem{cor}{Corollary}
\newtheorem{prop}{Proposition}
\newtheorem{defn}{Definition}
\newtheorem{rem}{Remark}

\DeclareMathOperator{\D}{\mathcal J}
\DeclareMathOperator{\jac}{jac}

\DeclareMathOperator{\HS}{\sf HS}
\DeclareMathOperator{\wHS}{\sf wHS}

\DeclareMathOperator{\mon}{Mon}
\DeclareMathOperator{\dreg}{\mathbb D_{reg}}

\DeclareMathOperator{\wdeg}{wdeg}
\DeclareMathOperator{\NF}{\sf NF}
\DeclareMathOperator{\LC}{\sf LC}
\DeclareMathOperator{\LM}{\sf LM}
\DeclareMathOperator{\Spol}{\sf Spol}
\DeclareMathOperator{\DEG}{\sf DEG}
\DeclareMathOperator{\Compl}{\sf Compl}

\begin{document}

\title{On the Complexity of the Generalized MinRank Problem}
 \author[$\dag$]{Jean-Charles Faug\`ere}
 \author[$\dag$]{Mohab Safey El Din}
 \author[*,$\dag$]{Pierre-Jean Spaenlehauer}

\affil[$\dag$]{Universit\'e Paris 6, 
INRIA Paris-Rocquencourt, PolSys Project, 
CNRS, UMR 7606
UFR Ing\'enierie 919, LIP6.\\
Case 169. 4, Place Jussieu, F-75252 Paris, France.}
\affil[*]{Computer Science Department, University of Western Ontario, London, ON, Canada.}

\date{}

\maketitle

\begin{abstract}
  We study the complexity of solving the \emph{generalized MinRank
    problem}, i.e. computing the set of points where the evaluation of
  a polynomial matrix has rank at most $r$. A natural algebraic
  representation of this problem gives rise to a \emph{determinantal
    ideal}: the ideal generated by all minors of size $r+1$ of the
  matrix. We give new complexity bounds for solving this problem using
  Gr\"obner bases algorithms under genericity assumptions on the input
  matrix. In particular, these complexity bounds allow us to identify
  families of generalized MinRank problems for which the arithmetic
  complexity of the solving process is polynomial in the number of
  solutions. We also provide an algorithm to compute a rational
  parametrization of the variety of a $0$-dimensional and radical
  system of bi-degree $(D,1)$. We show that its complexity can be
  bounded by using the complexity bounds for the generalized MinRank
  problem.
\end{abstract}


\section{Introduction}

We focus in this paper on the following problem:

\smallskip

\noindent \textbf{Generalized MinRank Problem}: given a field $\mathbb
K$, a $n\times m$ matrix $\mathcal M$ whose entries are polynomials of
degree $D$ in $\mathbb K[x_1,\ldots, x_k]$, and $r<\min(n,m)$ an
integer, compute the set of points at which the evaluation of
$\mathcal M$ has rank at most $r$.

\smallskip

This problem arises in many applications and this is what motivates
our study. In cryptology, the security of several multivariate
cryptosystems relies on the difficulty of solving the classical
MinRank problem (i.e. when the entries of the matrix are linear
\cite{KipSha99,FauLevPer08, BFP12}). In coding theory, rank-metric
codes can be decoded by computing the set of points where a polynomial
matrix has rank less than a given value
\cite{OurJoh02,FauLevPer08}. In non-linear computational geometry,
many incidence problems from enumerative geometry can be expressed by
constraints on the rank of a matrix whose entries are polynomials of
degree frequently larger than $1$ (see
e.g. \cite{macdonald2001common, sottile2003enumerative,
  sottile2002enumerative}). Also, in real geometry, optimization and
quantifier elimination \cite{safey2003polar, BanGiuHeiSafSch10,
  greuet2011global, HoSa12} the critical points of a map are defined
by the rank defect of its Jacobian matrix (whose entries have degrees
larger than $1$ most of the time in applications). Moreover, this
problem is also underlying other problems from symbolic computation
(for instance solving multi-homogeneous systems, see
e.g. \cite{FauSafSpa11}).

The ubiquity of this problem makes the development of algorithms
solving it and complexity estimates of first importance.  When
$\mathbb K$ is finite, the generalized MinRank problem is known to be
NP-complete \cite{BusFraSha99}; thus one can consider this problem as
a hard problem.

\smallskip

To study the Generalized MinRank problem, we consider the algebraic
system of all the $(r+1)$-minors of the input matrix. Indeed, these
minors simultaneously vanish on the locus of rank defect and hence
give rise to a section of a \emph{determinantal ideal}.

\smallskip

Several solving tools can be used to solve this algebraic system by
taking profit of the underlying structure. For instance, the \emph{geometric
resolution} in \cite{GiuLecSal01} can use the fact that these systems can
be evaluated efficiently. Also, recent works on homotopy methods \cite{Ver99} show
that numerical algorithms can solve determinantal
problems.

\smallskip

In this paper, we focus on Gr\"obner bases algorithms. 
A representation of the locus of rank defect is obtained by computing
a lexicographical Gr\"obner basis by using the algorithms $F_5$
\cite{Fau02} and FGLM \cite{FauGiaLazMor93}. 
Indeed, experiments suggest that these algorithms take profit of the
determinantal structure. The aim of this work is to give an explanation
of this behavior from the viewpoint of asymptotic complexity analysis.

\subsubsection*{Related works}
An important related theoretical issue is to understand the algebraic
structure of the ideal $\D_r\subset \mathbb K[U]$ (where $U$
is the set of variables $\{u_{1,1},\ldots, u_{n,m}\}$) generated by
the $(r+1)$-minors of the matrix:
$$\mathcal U=\begin{pmatrix}
u_{1,1}&\dots&u_{1,m}\\
\vdots&\ddots&\vdots\\
u_{n,1}&\dots&u_{n,m}\\
\end{pmatrix}.$$ The ideal $\D_r$ has been extensively studied
during last decades. In particular, explicit formulas for its
degree and for its Hilbert series are known (see e.g. \cite[Example
14.4.14]{Ful97} and \cite{ConHer94}), as well as structural properties
such as Cohen-Macaulayness and primality \cite{HocEag70,HocEag71}.

In cryptology, \cite{KipSha99} have proposed a multi-homogeneous
algebraic modeling which can be seen as a generalization of the
Lagrange multipliers  and is designed as follows: a
polynomial $n\times m$ matrix $\mathcal M\in \mathbb K[X]^{n\times m}$
(where $X$ denotes the set of variables $\{x_1,\ldots, x_k\}$) has
rank at most $r$ if and only if the dimension of its right kernel
is greater than $m-r-1$. Consequently, by introducing $r(m-r)$ fresh
variables $y_{1,1},\ldots, y_{r,m-r}$, we can consider the system of
bi-degree $(D,1)$ in $\mathbb K[x_1,\ldots, x_k,y_{1,1},\ldots,
y_{r,m-r}]$ defined by

$$\mathcal M\cdot \begin{pmatrix}
1&0&\dots&0\\
0&1&\dots&0\\
\vdots&\ddots&\ddots&\vdots\\
0&0&\dots&1\\
y_{1,1}&y_{1,2}&\dots&y_{1,m-r}\\
\vdots&\vdots&\ddots&\vdots\\
y_{r,1}&y_{r,2}&\dots&y_{r,m-r}
\end{pmatrix} = 0.$$

If $(x_1,\ldots,x_k,y_{1,1},\ldots,y_{r, m-r})$ is a solution of that
system, then the evaluation of the matrix $\mathcal M$ at the point
$(x_1,\ldots, x_k)$ has rank at most $r$.

In \cite{FauSafSpa10a}, the case of square linear matrices is studied
by performing a complexity analysis of the Gr\"obner bases
computations. In particular, this
investigation showed that the overall complexity is polynomial in the
size of the matrix when the rank defect $n-r$ is constant. This
theoretical analysis is supported by experimental results.  The proofs
were complete when the system has positive dimension, but depended on
a variant of a conjecture by Fr\"oberg in the $0$-dimensional case.

\subsubsection*{Main results}
We generalize in several ways the results from \cite{FauSafSpa10a}
where only the case of square linear matrices was investigated: our
contributions are the following.
\begin{itemize}
\item We deal with non-square matrices whose entries are polynomials
  of degree $D$ with generic coefficients; this is achieved by using
  more general tools than those considered in \cite{FauSafSpa10a}
  (weighted Hilbert series).  This generalization is important for
  applications in geometry and optimization for instance.
\item When $n=(p-r) (q-r)$, the solution set of the generalized MinRank problem
   has dimension $0$. In that case, our proofs in this paper do not rely on
   Fr\"oberg's conjecture; this has been achieved by modifying our
   proof techniques and using more sophisticated and structural   properties of determinantal ideals. This is important for
   applications in cryptology (see e.g. the sets of parameters A, B and C in the MinRank authentication scheme \cite{Cou01}). 
\end{itemize}

Our results are complexity bounds for Gr\"obner bases algorithms when the
input system is the set of $(r+1)$-minors of a $n\times m$ matrix
$\mathcal M$, whose entries are polynomials of degree $D$ with generic
coefficients.

By generic, we mean that there exists a non-identically
null multivariate polynomial $h$ such that the complexity results hold
when this polynomial does not vanish on the coefficients of the
polynomials in the matrix. Therefore, from a practical viewpoint, the
complexity bounds can be used for applications where the base field
$\mathbb K$ is large enough: in that case, the probability that the
coefficients of $\mathcal M$ do not belong to the zero set of $h$ is
close to $1$.

We start by studying the homogeneous generalized MinRank problem
(i.e. when the entries of $\mathcal M$ are homogeneous polynomials)
and by proving an explicit formula for the Hilbert series of the ideal
$\mathcal I_r$ generated by the $(r+1)$-minors of the matrix $\mathcal
M$.  The general framework of the proofs is the following: we consider
the ideal $\D_r\subset \mathbb K[U]$ generated by the
$(r+1)$-minors of a matrix $\mathcal U=(u_{i,j})$ whose entries are
variables. Then we consider the ideal $\widetilde{\D_r}=\D_r+\langle g_1,\ldots,g_{nm}\rangle\subset\mathbb
K[U,X]$, where the polynomials $g_i$ are quasi-homogeneous forms that
are the sum of a linear form in $\mathbb K[U]$ and of a homogeneous
polynomial of degree $D$ in $\mathbb K[X]$. If some conditions on the
$g_i$ are verified, by performing a linear combination of the
generators there exists $f_{1,1},\ldots, f_{n,m}\in \mathbb K[X]$ such
that $$\widetilde{\D_r} =\D_r+\langle
u_{1,1}-f_{1,1},\ldots, u_{n,m}-f_{n,m}\rangle.$$ Then we use the fact
that $\left(\D_r+\langle u_{1,1}-f_{1,1},\ldots,
  u_{n,m}-f_{n,m}\rangle\right)\cap \mathbb K[X] = \mathcal I_r$ to
prove that properties of generic quasi-homogeneous sections of
$\D_r$ transfer to $\mathcal I_r$ when the entries of the
matrix $\mathcal M$ are generic. This allows us to use results known
about the ideal $\D_r$ to study the algebraic structure of
$\mathcal I_r$.

\smallskip

We study separately three different cases:
\begin{itemize}
\item $k>(n-r)(m-r)$. Under genericity assumptions on the input, the
  solutions of the generalized MinRank problem are an algebraic
  variety of positive dimension. Recall that the complexity results
  were only proven for $D=1$ and $n=m$ in \cite{FauSafSpa10a}. We
  generalize here for any $D\in \mathbb N$.
\item $k=(n-r)(m-r)$. This is the $0-dimensional$ case, where the
  problem has finitely-many solutions under genericity
  assumptions. Recall that the results in \cite{FauSafSpa10a} were
  only stated for $D=1$ and $n=m$, and they depended on a variant of
  Fr\"oberg's conjecture. In this paper, we give complete proofs for
  $D\in \mathbb N$ which do not rely on any conjecture.
\item $k<(n-r)(m-r)$. In the over-determined case, we still need to
  assume a variant of Fr\"oberg's conjecture to generalize the results
  in \cite{FauSafSpa10a}.
\end{itemize}

In particular, we prove that, for $k\geq (n-r)(m-r)$, the Hilbert
series of $\mathcal I_r$ is the power series
expansion of the rational function
$$\HS_{\mathcal I_r}(t)=\frac{\det A_r(t^D) (1-t^D)^{(n-r)(m-r)}}{t^{D\binom{r}{2}}
  (1-t)^{k}},$$ where $A_r(t)$ is the $r\times
r$ matrix whose $(i,j)$-entry is $\sum_k \binom{m-i}{k} \binom{n-j}{k}
t^k$.  Assuming w.l.o.g. that $m\leq n$, we also prove that the degree
of $\mathcal I_r$ is equal to
$$\DEG(\mathcal I_r)=D^{(n-r)(m-r)} \prod_{i=0}^{m-r-1}\frac{i! (n+i)!}{(m-1-i)! (n-r+i)!}.$$

These explicit formulas permit to derive complexity bounds on the
complexity of the problem. Indeed, one way to get a representation of
the solutions of the problem in the $0$-dimensional case is to compute
a \emph{lexicographical} Gr\"obner basis of the ideal generated by the
polynomials. This can be achieved by using first the $F_5$ algorithm
\cite{Fau02} to compute a Gr\"obner basis for the so-called
\emph{grevlex} ordering and then use the FGLM algorithm
\cite{FauGiaLazMor93} to convert it into a \emph{lexicographical}
Gr\"obner basis. The complexities of these algorithms are governed by
the degree of regularity and by the degree of the ideal.

Therefore the theoretical results on the structure of $\mathcal I_r$
yield bounds on the complexity of solving the generalized MinRank
problem with Gr\"obner bases algorithms. More specifically, when
$k=(n-r)(m-r)$ and under genericity assumptions on the input
polynomial matrix, we prove that the arithmetic complexity for
computing a lexicographical Gr\"obner basis of $\mathcal I_r$ is upper bounded by
$$O\left(\binom{n}{r+1}\binom{m}{r+1}\binom{\dreg+k}{k}^\omega + k\left(\DEG\left(\mathcal I_r\right)\right)^3\right),$$
where $2\leq\omega\leq 3$ is a feasible exponent for the matrix multiplication, and
$$\dreg=D r (m-r) + (D-1) k +1.$$

This complexity bound permits to identify families of Generalized
MinRank problems for which the number of arithmetic operations during
the Gr\"obner basis computations is polynomial in the number of
solutions.

In the over-determined case (i.e. $k<(n-r)(m-r)$), we obtain
similar complexity results, by assuming a variant of Fr\"oberg's
conjecture which is supported by experiments.

Finally, we show that complexity bounds for solving systems of
bi-degree $(D,1)$ can be obtained from these results on the
generalized MinRank problem. We give an algorithm whose arithmetic complexity is upper
bounded by
$$O\left(\binom{n_x+n_y}{n_y+1}\binom{D(n_x+n_y)+1}{n_x}^\omega + n_x\left(D^{n_x}\binom{n_x+n_y}{n_x}\right)^3\right),$$
for solving systems of $n_x+n_y$ equations of bi-degree $(D,1)$ in
$\mathbb K[x_1,\ldots, x_{n_x},y_1,\ldots, y_{n_y}]$ which are radical and $0$-dimensional.

\subsubsection*{Organization of the paper}
Section \ref{sec:notations} provides notations used throughout this
paper and preliminary results. In Section \ref{sec:transfer}, we show
how properties of the ideal $\D_r$ generated by the $(r+1)$-minors of
$\mathcal U$ transfer to the ideal $\mathcal I_r$. Then, the case when
the homogeneous Generalized MinRank Problem has non-trivial solutions
(under genericity assumptions) is studied in Section
\ref{sec:welldef}. Section \ref{sec:overdef} is devoted to the study
of the over-determined MinRank Problem (i.e. when
$k<(n-r)(m-r)$). Then, the complexity analysis is performed in
Section~\ref{sec:compl}. Some consequences of this complexity analysis
are drawn in Section \ref{sec:cases}. Experimental results are given
in Section \ref{sec:expe} and applications to the complexity of
solving bi-homogeneous systems of bi-degree $(D,1)$ are investigated
in Section \ref{sec:bihomogeneous}.

\section{Notations and preliminaries}
\label{sec:notations}
Let $\mathbb K$ be a field and $\overline{\mathbb K}$ be its algebraic closure.
In the sequel, $n$, $m$, $r$ and $k$ and $D$ are positive integers
with $r<m\leq n$. For $d\in \mathbb N$, $\mon(d,k)$ denotes the set of
monomials of degree $d$ in the polynomial ring $\mathbb K[x_1,\ldots,
x_k]$. Its cardinality is $\#\mon(d,k)=\binom{d-1+k}{d}.$ 

We denote by $\mathfrak a$ the set of parameters $\{\mathfrak
a_{t}^{(i,j)} : 1\leq i\leq n, 1\leq j\leq m, t\in \mon(D,
k)\}$. The set of variables $\{u_{i,j} : 1\leq i\leq n, 1\leq j\leq
m\}$ (resp. $\{x_1,\ldots, x_k\}$) is denoted by $U$ (resp. $X$).

For $1\leq i\leq n, 1\leq j\leq m$, we denote by $f_{i,j}\in \mathbb
K(\mathfrak a)[X]$ a generic form of degree $D$
$$f_{i,j} = \sum_{t\in \mon(D,k)} \mathfrak a_{t}^{(i,j)} t.$$

Let $\mathcal I_r\subset\mathbb K(\mathfrak a)[X]$ be the ideal
generated by the $(r+1)$-minors of the $n\times m$ matrix
$$\mathcal M=\begin{pmatrix}
  f_{1,1}&\dots&f_{1,m}\\ \vdots&\ddots&\vdots\\ f_{n,1}&\dots&f_{n,m}
\end{pmatrix},$$
and $\D_r\subset\mathbb K(\mathfrak a)[U,X]$ be the
determinantal ideal generated by the $(r+1)$-minors of the matrix
$$\mathcal U=\begin{pmatrix}
u_{1,1}&\dots&u_{1,m}\\
\vdots&\ddots&\vdots\\
u_{n,1}&\dots&u_{n,m}\\
\end{pmatrix}.$$

We define $\widetilde{\mathcal I_r}$ as the ideal $\D_r +
\langle u_{i,j}-f_{i,j}\rangle_{1\leq i\leq n,1\leq j\leq m} \subset
\mathbb K(\mathfrak a)[U,X]$. Notice that $\widetilde{\mathcal I_r}=\mathcal I_r + \langle u_{i,j}-f_{i,j}\rangle_{1\leq
  i\leq n,1\leq j\leq m} \subset \mathbb K(\mathfrak
a)[U,X]$. Therefore, $\mathcal I_r=\widetilde{\mathcal
  I_r}\cap\mathbb K(\mathfrak a)[X]$.

\smallskip

By slight abuse of notation, if $I$ is a proper homogeneous ideal of a polynomial ring $\mathbb K[X]$, we call \emph{Hilbert series} of $I$ and we note $\HS_I\in\mathbb Z[[t]]$ the Hilbert series of its quotient algebra $\mathbb K[X]/I$ with the grading defined by $\deg(x_i)=1$ for all~$i$:
$$\HS_I(t)=\sum_{d\geq 0} \dim_{\mathbb K}\left(\mathbb K[X]_d/I_d\right)t^d,$$
where $\mathbb K[X]_d$ denotes the vector space of homogeneous polynomials of degree $d$ and $I_d=I\cap \mathbb K[X]_d$. 

We call \emph{dimension} of $I$ the Krull dimension of the quotient ring $\mathbb K[X]/I$.

\subsubsection*{Quasi-homogeneous polynomials.}
We need to balance the degrees of the entries of the matrix $\mathcal
U$ with the degrees of the entries of $\mathcal M$. This can be
achieved by putting a \emph{weight} on the variables $u_{i,j}$, giving
rise to \emph{quasi-homogeneous} polynomials. A polynomial $f\in
\mathbb K[U,X]$ is called \emph{quasi-homogeneous} (of type $(D,1)$)
if the following condition holds (see e.g. \cite[Definition 2.11, page
120]{GreLosShu07}):
$$f(\lambda^D u_{1,1},\ldots,\lambda^D u_{n,m},\lambda x_1,\ldots, \lambda x_k)=\lambda^d f(u_{1,1},\ldots,u_{n,m},x_1,\ldots, x_k).$$
The integer $d$ is called the weight degree of $f$ and denoted by
$\wdeg(f)$.

An ideal $I\subset\mathbb K[U,X]$ is called quasi-homogeneous (of type
$(D,1)$) if there exists a set of quasi-homogeneous generators. In
this case, we denote by $\mathbb K[U,X]_d$ the $\mathbb K$-vector
space of quasi-homogeneous polynomials of weight degree $d$, and $I_d$
denote the set $\mathbb K[U,X]_d\cap I$. 

\begin{prop}
Let $I\subset\mathbb K[U,X]$ be an ideal.  Then the following statements are equivalent:
\begin{enumerate}
\item there exists a set of quasi-homogeneous generators of $I$;
\item the sets $I_d$ are subspaces of $\mathbb K[U,X]_d$, and $I=\bigoplus_{d\in \mathbb N} I_d$.
\end{enumerate}
\end{prop}

\begin{proof} See e.g. \cite[Chapter 8]{MilStu05}.
\end{proof}

If $I$ is a quasi-homogeneous ideal, then its \emph{weighted Hilbert
series} $\wHS_I(t)\in \mathbb Z[[t]]$ is defined as follows:
$$\wHS_I(t)=\sum_{d\in \mathbb N} \dim(\mathbb K[U,X]_d/I_d) t^d.$$

\section{Transferring properties from $\D_r$ to $\mathcal I_r$}
\label{sec:transfer}
In this section, we prove that generic structural properties (such as
the dimension, the structure of the leading monomial ideal,\ldots) of
the ideal $\widetilde{\mathcal I_r}$ are the same as properties of the
ideal $\D_r$ where several generic forms have been
added. Hence several classical properties of the determinantal ideal
$\D_r$ transfer to the ideal $\widetilde{\mathcal I_r}$. For
instance, this technique permits to obtain explicit forms of the
Hilbert series of the ideal $\widetilde{\mathcal I_r}$.

In the following, we denote by $\mathfrak b$ and $\mathfrak c$ the following sets of parameters:
$$\begin{array}{rcl}
\mathfrak b&=&\{\mathfrak b^{(\ell)}_{t} \mid t\in \mon(D,k), 1\leq\ell\leq nm\};\\
\mathfrak c&=&\{\mathfrak c^{(\ell)}_{i,j} \mid 1\leq i\leq n, 1\leq j\leq m, 1\leq\ell\leq nm\}.
\end{array}$$

Also, $g_1,\ldots, g_{nm}\in\mathbb
K(\mathfrak b,\mathfrak c)[U,X]$ are generic quasi-homogeneous forms of type $(D,1)$ and of
weight degree $D$:
$$g_\ell=\sum_{t\in\mon(D,k)}\mathfrak b_{t}^{(\ell)} t + \sum_{\substack{1\leq i\leq n\\1\leq j\leq m}}\mathfrak c_{i,j}^{(\ell)} u_{i,j}.$$

We let $\widetilde{\D_r}$ denote the ideal $\D_r +
\langle g_1,\ldots, g_{nm}\rangle \subset \mathbb K(\mathfrak
b,\mathfrak c)[U,X]$.  Here and subsequently, for $\mathbf
a=(a_{i,j})\in \overline{\mathbb K}^{n m\binom{D-1+k}{D}}$, we denote by
$\varphi_{\mathbf a}$ the following evaluation morphism:
$$\begin{array}{cccc}
\varphi_{\bf a} :&\mathbb K[\mathfrak a]&\longrightarrow&\overline{\mathbb K}\\
&f(\mathfrak a_{1,1},\ldots, \mathfrak a_{n,m})&\longmapsto&f( a_{1,1},\ldots, a_{n,m})
\end{array}$$
Also, for $(\mathbf b, \mathbf c)\in \overline{\mathbb K}^{nm\left(\binom{D-1+k}{D}+nm\right)}$, we denote by $\psi_{\mathbf b,\mathbf c}$ the evaluation morphism:
$$\begin{array}{cccc}
\psi_{\mathbf b,\mathbf c} :&\mathbb K[\mathfrak b,\mathfrak c]&\longrightarrow&\overline{\mathbb K}\\
&f(\mathfrak b,\mathfrak c)&\longmapsto&f(\mathbf b,\mathbf c)
\end{array}$$

By abuse of notation, we let $\varphi_{\mathbf
  a}(\widetilde{\mathcal I_r})$ (resp. $\psi_{\mathbf b, \mathbf
  c}(\widetilde{\D_r})$) denote the ideal $\D_r+\langle u_{i,j}-\varphi_{\bf a}(f_{i,j}) \rangle\subset \overline{\mathbb
K}[U,X]$ (resp. $\D_r+\left\langle \psi_{\mathbf b, \mathbf
  c}(g_1),\ldots, \psi_{\mathbf b, \mathbf c}(g_{nm})\right\rangle\subset\overline{\mathbb K}[U,X]$).

We call \emph{property} a map from the set of ideals of $\overline{\mathbb
K}[U,X]$ to $\{{\tt true}, {\tt
  false}\}$: $$\begin{array}{rrcl}\mathcal P:&\mathsf{Ideals}(\overline{\mathbb
  K}[U,X])&\rightarrow&\{{\tt true},{\tt false}\}\end{array}.$$

\begin{defn}
Let  $\mathcal P$ be a property. We say that $\mathcal P$ is 
\begin{itemize}
\item $\widetilde{\mathcal I_r}$-generic if there exists a non-empty Zariski open subset $O\subset \overline{\mathbb K}^{n m\binom{D-1+k}{D}}$ such that
  $$\mathbf a\in O\Rightarrow \mathcal P\left(\varphi_{\mathbf a}\left(\widetilde{\mathcal
      I_r}\right)\right)={\tt true};$$
\item $\widetilde{\D_r}$-generic if there exists a non-empty Zariski open subset $O\subset\overline{\mathbb
    K}^{nm\left(\binom{D-1+k}{D}+nm\right)}$ such that 
$$(\mathbf b,\mathbf c)\in O\Rightarrow \mathcal P\left(\psi_{\mathbf b,\mathbf
      c}\left(\widetilde{\D_r}\right)\right)={\tt true}.$$
\end{itemize}
\end{defn}

The following lemma is the main result of this section:

\begin{lem}\label{lem:transfer}
A property $\mathcal P$ is $\widetilde{\mathcal I_r}$-generic if and only if it is $\widetilde{\D_r}$-generic.
\end{lem}

\begin{proof}
  To obtain a representation of $\varphi_{\bf
    a}\left(\widetilde{\D_r}\right)$ for a generic $\mathbf a$
  as a specialization of $\widetilde{\mathcal I_r}$ (and conversely),
  it is sufficient to perform a linear combination of the
  generators. The point of this proof is to show that genericity
  is preserved during this linear transform.

  In the sequel we denote by $\mathfrak A, \mathfrak B$ and $\mathfrak C$ the following matrices (of respective sizes $nm \times \binom{D-1+k}{D}$, $nm \times \binom{D-1+k}{D}$ and $nm\times nm$):
$$\begin{array}{rcl}\mathfrak A&=&\begin{pmatrix} \mathfrak a^{(1)}_{x_1^D}&\mathfrak a^{(1)}_{x_1^{D-1}x_2}&\dots&\mathfrak a^{(1)}_{x_k^D}\\
  \vdots&\vdots&\vdots&\vdots\\
  \mathfrak a^{(n m)}_{x_1^D}&\mathfrak a^{(n
    m)}_{x_1^{D-1}x_2}&\dots&\mathfrak a^{(n m)}_{x_k^D}
\end{pmatrix}\\\mathfrak B&=&\begin{pmatrix} \mathfrak b^{(1)}_{x_1^D}&\mathfrak b^{(1)}_{x_1^{D-1}x_2}&\dots&\mathfrak b^{(1)}_{x_k^D}\\
  \vdots&\vdots&\vdots&\vdots\\
  \mathfrak b^{(n m)}_{x_1^D}&\mathfrak b^{(n
    m)}_{x_1^{D-1}x_2}&\dots&\mathfrak b^{(n m)}_{x_k^D}
\end{pmatrix}\\ \mathfrak C &=& \begin{pmatrix}
  \mathfrak c^{(1)}_{1,1}&\dots&\mathfrak c^{(1)}_{n,m}\\
  \vdots&\vdots&\vdots\\
  \mathfrak c^{(n m)}_{1,1}&\dots&\mathfrak c^{(n m)}_{n,m}
\end{pmatrix}.\end{array}$$
Therefore, we have
$$\begin{array}{rcl}
\begin{pmatrix}u_{1,1}-f_{1,1}\\\vdots\\u_{n,m}-f_{n,m}\end{pmatrix}&=&\mathsf{Id}_{nm}\cdot\begin{pmatrix}u_{1,1}\\\vdots\\u_{n,m}\end{pmatrix}-\mathfrak A \cdot \begin{pmatrix}x_1^D\\x_1^{D-1}x_2\\\vdots\\x_k^D\end{pmatrix}\\
\begin{pmatrix}g_1\\\vdots\\g_{nm}\end{pmatrix}&=&\mathfrak C\cdot\begin{pmatrix}u_{1,1}\\\vdots\\u_{n,m}\end{pmatrix}+\mathfrak B \cdot \begin{pmatrix}x_1^D\\x_1^{D-1}x_2\\\vdots\\x_k^D\end{pmatrix}
\end{array}$$

In this proof, for $\mathbf a\in \mathbb K^{nm \binom{D-1+k}{D}}$
(resp. $\mathbf b\in \mathbb K^{nm \binom{D-1+k}{D}}, \mathbf c\in
\mathbb K^{n^2m^2}$), the notation $\mathbf A$ (resp. $\mathbf B,
\mathbf C$) stands for the evaluation of the matrix $\mathfrak A$
(resp. $\mathfrak B, \mathfrak C$) at $\mathbf a$ (resp. $\mathbf b,
\mathbf c$). Also, we implicitly identify $\mathbf A$ with $\mathbf a$
(resp. $\mathbf B$ with $b$, $\mathbf C$ with $\mathbf c$, $\mathfrak A$ with $\mathfrak a$, $\mathfrak B$ with $\mathfrak b$, $\mathfrak C$ with $\mathfrak c$).

\begin{itemize}
\item Let $\mathcal P$ be a $\widetilde{\mathcal I_r}$-generic
  property. Thus there exists a non-zero polynomial $h_1(\mathfrak A)\in \overline{\mathbb
    K}[\mathfrak a]$ such that if $h_1(\mathbf A)\neq 0$ then
  $\mathcal P\left(\varphi_{\mathbf a}(\widetilde{\mathcal
      I_r})\right)={\tt true}$. 

  Let $\mathsf{adj}(\mathfrak C)$ denote the adjugate of $\mathfrak C$
  (i.e. $\mathsf{adj}(\mathfrak C)=\det(\mathfrak C)\cdot \mathfrak
  C^{-1}$ in $\mathbb K(\mathfrak c)$). Consider the polynomial $\widetilde{h_1}$ defined by 
  $\widetilde{h_1}(\mathfrak B, \mathfrak
  C)=h_1(-\mathsf{adj}(\mathfrak C)\cdot \mathfrak B)\in
  \overline{\mathbb K}[\mathfrak b,\mathfrak c]$. The polynomial inequality
  $\det(\mathfrak C)\widetilde{h_1}(\mathfrak B, \mathfrak C)\neq 0$
  defines a non-empty Zariski open subset $O\subset\overline{\mathbb
    K}^{nm\left(\binom{D-1+k}{D}+nm\right)}.$ Let $(\mathbf B, \mathbf
  C)\in O$ be an element in this set, then $\mathbf C$ is invertible
  since $\det(\mathbf C)\neq 0$. Let $\widetilde{\mathbf A}$ be the
  matrix $\widetilde{\mathbf A}=-\mathsf{adj}(\mathbf C)\cdot \mathbf
  B$. Therefore the generators of the ideal $\varphi_{\widetilde{\mathbf
      a}}\left(\widetilde{\mathcal I_r}\right)$ are an invertible
  linear combination of the generators of $\psi_{\mathbf b,\mathbf
    c}\left(\widetilde{\D_r}\right)$. Consequently,
  $\varphi_{\widetilde{\mathbf a}}\left(\widetilde{\mathcal
      I_r}\right) = \psi_{\mathbf b,\mathbf
    c}\left(\widetilde{\D_r}\right)$. Moreover,
  $h_1(\widetilde{\bf A})=\widetilde{h_1}(\mathbf B,\mathbf C)\neq 0$
  implies that the polynomial $\widetilde{h_1}$ is not identically $0$. Therefore,
$$\forall (\mathbf b, \mathbf c)\in O, \mathcal P\left( \psi_{\mathbf
    b,\mathbf c}\left(\widetilde{\D_r}\right)\right)=\mathcal
P\left(\varphi_{\widetilde{\mathbf a}}\left(\widetilde{\mathcal
      I_r}\right)\right)={\tt true},$$ and hence $\mathcal P$ is a
$\widetilde{\D_r}$-generic property.

\item Conversely, consider a $\widetilde{\D_r}$-generic
  property $\mathcal P$. Thus, there exists a non-zero polynomial $h_2(\mathfrak
  B,\mathfrak C)\in \overline{\mathbb K}[\mathfrak b,\mathfrak c]$ such that if
  $h_2(\mathbf b,\mathbf c)\neq 0$ then $\mathcal P\left(\psi_{\mathbf b,\mathbf c}(\widetilde{\D_r})\right)={\tt true}$.  Since
  $\mathcal P$ is $\widetilde{\D_r}$-generic, there exists
  $(\mathbf b, \mathbf c)$ such that $h_2(\mathbf b, \mathbf
  c)\det(\mathbf c)\neq 0$. 
Let $\widetilde{h_2}$ be the polynomial 
  $\widetilde{h_2}(\mathfrak b)=h_2(-\mathbf C\cdot\mathfrak B, \mathbf C)$.

  Since $\det(\mathbf C)\neq 0$, the matrix $\mathbf C$ is invertible
  and $\widetilde{h_2}(-\mathbf C^{-1}\cdot\mathbf B)=h_2(\mathbf
  B,\mathbf C)\neq 0$ and hence the polynomial $\widetilde{h_2}$ is not
  identically $0$. Moreover, if $\mathbf a\in \mathbb K^{n
    m\binom{D-1+k}{D}}$ is such that $\widetilde{h_2}(\mathbf A)\neq
  0$, then $h_2(-\mathbf C\cdot\mathbf A,\mathbf C)\neq 0$ and thus
  $\mathcal P\left(\psi_{-\mathbf C\cdot\mathbf A,\mathbf
      C}(\widetilde{\D_r})\right)={\tt true}$. Finally,
  $\psi_{-\mathbf C\cdot\mathbf A,\mathbf C}(\widetilde{\D_r})=\varphi_{\mathbf A}(\widetilde{\mathcal I_r})$ since the
  generators of $\psi_{-\mathbf C\cdot\mathbf A,\mathbf
    C}(\widetilde{\D_r})$ are an invertible linear combination
  of that of $\varphi_{\mathbf a}(\widetilde{\mathcal I_r})$ (the
  linear transformation being given by the invertible matrix $\mathbf
  C$) and hence they generate the same ideal. Therefore, the property
  $\mathcal P$ is $\widetilde{\mathcal I_r}$-generic.
\end{itemize}
\end{proof}

In the sequel, $\prec$ is an admissible monomial ordering (see e.g \cite[Chapter 2, \S 2, Definition 1]{CoxLitShe97}) on $\mathbb K[U,X]$, and for any polynomial $f\in\mathbb K[U,X]$, $\LM(f)$ denotes its leading monomial with respect to $\prec$. If $I$ is an ideal of $\mathbb K[U,X]$, $\mathbb K(\mathfrak a)[U,X]$, or $\mathbb K(\mathfrak b, \mathfrak c)[U,X]$, we let $\LM(I)$ denote the ideal generated by the leading monomials of the polynomials. 

By slight abuse of notation, if $I_1$ and $I_2$ are ideals of $\mathbb K[U,X]$, $\mathbb K(\mathfrak a)[U,X]$, or $\mathbb K(\mathfrak b, \mathfrak c)[U,X]$ ($I_1$ and $I_2$ are not necessarily ideals of the same ring), we write $\LM(I_1)=\LM(I_2)$ if the sets $\{\LM(f)\mid f\in I_1\}$ and $\{\LM(f)\mid f\in I_2\}$ are equal.

\begin{lem}\label{lem:LMgen}
Let $\mathcal P_{\widetilde{\mathcal I_r}}$ and $\mathcal P_{\widetilde{\D_r}}$ be the properties defined by
$$\begin{array}{c}\mathcal P_{\widetilde{\mathcal I_r}}(I)=\begin{cases}{\tt true}\text{ if }\LM (I)=\LM\left(\widetilde{\mathcal I_r}\right);\\
{\tt false} \text{ otherwise.}\end{cases}\\
\mathcal P_{\widetilde{\D_r}}(I)=\begin{cases}{\tt true}\text{ if }\LM (I)=\LM\left(\widetilde{\D_r}\right);\\
{\tt false} \text{ otherwise.}\end{cases}\end{array}$$
Then $\mathcal P_{\widetilde{\mathcal I_r}}$ (resp. $\mathcal P_{\widetilde{\D_r}}$) is a $\widetilde{\mathcal I_r}$-generic (resp. $\widetilde{\D_r}$-generic) property.
\end{lem}

\begin{proof}
  We prove here that $\mathcal P_{\widetilde{\mathcal I_r}}$ is
  $\widetilde{\mathcal I_r}$-generic (the proof for $\mathcal
  P_{\widetilde{\D_r}}$ is similar).

\smallskip

The outline of this proof is the following: during the computation of
a Gr\"obner basis $G$ of $\widetilde{\mathcal I_r}$ in $\mathbb
K(\mathfrak a)[U,X]$ (for instance with Buchberger's algorithm), a
finite number of polynomials are constructed. Let $\varphi_{\mathbf
  a}$ be a specialization. If the images by $\varphi_{\mathbf a}$ of
the leading coefficients of all non-zero polynomials arising during
the computation do not vanish, then $\varphi_{\mathbf
  a}(G)\subset\varphi_{\mathbf a}(\widetilde{\mathcal I_r})$ is a
Gr\"obner basis of the ideal it generates. It remains to prove that
$\varphi_{\mathbf a}(G)$ is a Gr\"obner basis of $\varphi_{\mathbf
  a}(\widetilde{\mathcal I_r})$. This is achieved by showing that
generically, the normal form (with respect to $\varphi_{\mathbf
  a}(G)$) of the generators of $\varphi_{\mathbf
  a}(\widetilde{\mathcal I_r})$ is equal to zero.
  
\smallskip

For polynomials $f_1,f_2$, we let $\LC(f_1)$ (resp. $\LC(f_2)$) denote
the leading coefficient of $f_1$ (resp. $f_2$) and $\Spol(f_1,
f_2)=\frac{\mathsf{LCM}(\LM(f_1),
  \LM(f_2))}{\LC(f_1)\LM(f_1)}f_1-\frac{\mathsf{LCM}(\LM(f_1),
  \LM(f_2))}{\LC(f_2)\LM(f_2)}f_2$ denote the
\emph{S-polynomial} of $f_1$ and $f_2$.

We need to prove that there exists a non-empty Zariski open subset
$O_1\subset \overline{\mathbb K}^{n m \binom{D-1+k}{D}}$ such that
$$\mathbf a\in O_1\Rightarrow \LM(\varphi_{\bf a}(\widetilde{\mathcal
  I_r})) = \LM(\widetilde{\mathcal I_r}).$$ To do so, consider a
Gr\"obner basis $G\subset \mathbb K(\mathfrak a)[U,X]$ of
$\widetilde{\mathcal I_r}$ such that each polynomial $g$ can be
written as a combination $g=\sum h_\ell f_\ell$, where the $f_\ell$'s
range over the set of minors of size $r+1$ of $\mathcal U$ and the
polynomials $u_{i,j}-f_{i,j}$, and $h_\ell\in \mathbb K[\mathfrak
a][U,X]$. Buchberger's criterion states that S-polynomials of
polynomials in a Gr\"obner basis reduce to zero \cite[Chapter 2, \S 6,
Theorem 6]{CoxLitShe97}. Thus each S-polynomial of $g_i, g_j\in G$ can be
rewritten as an algebraic combination
$$\Spol(g_i,g_j)= \sum_{\ell} h_\ell' g_\ell,$$ where the polynomials $h_\ell'$
belongs to $\mathbb K(\mathfrak a)[U,X]$ and such that $\{g_1,\ldots,
g_{t_{i,j}}\}\subset G$ and for each $1\leq s\leq t_{i,j}$, $\LM(g_s)$
divides $\LM(\Spol(g,g')- \sum_{\ell=1}^{s-1} h_\ell' g_\ell)$. Next,
consider:
\begin{itemize}
\item the product $Q_1(\mathfrak a)=\prod_{g\in G}\LC(g)$ of the leading
coefficients of the polynomials in the Gr\"obner basis;
\item for all $(g_i,g_j)\in G^2$ such that $\Spol(g_i,g_j)\neq 0$, the
  product $Q_2(\mathfrak a)$ of the numerators and denominators of the
  leading coefficients arising during the reduction of $\Spol(g_i,
  g_j)$.
\end{itemize}

These coefficients belongs to $\mathbb K[\mathfrak a]$. Denote by
$Q(\mathfrak a)=Q_1(\mathfrak a) Q_2(\mathfrak a)\in \mathbb
K[\mathfrak a]$ their product. The inequality $Q(\mathfrak a)\neq 0$
defines a non-empty Zariski open subset $O_1\subset\overline{\mathbb
  K}^{ n m \binom{D-1+k}{D}}$. If $\mathbf a\in O_1$, then
$$\varphi_{\bf a}(\Spol(g,g'))=\sum_{\ell=1}^t \varphi_{\bf a}(h_\ell') \varphi_{\bf a}(g_\ell),$$
and for each $1\leq i\leq t$, $\LM(\varphi_{\bf a}(g_i))$ divides
$\LM(\varphi_{\bf a}(\Spol(g,g'))- \sum_{\ell=1}^{i-1} \varphi_{\bf
  a}(h_\ell') \varphi_{\bf a}(g_\ell))$. Thus $\varphi_{\bf a}(G)$ is a
Gr\"obner basis of the ideal it spans. Moreover, $\langle \varphi_{\bf
  a}(G)\rangle\subset\varphi_{\bf a}(\widetilde{\mathcal I_r})$. 

We prove now that there exists a non-empty Zariski open set where the
other inclusion $\varphi_{\bf a}(\widetilde{\mathcal
  I_r})\subset\langle \varphi_{\bf a}(G)\rangle$ holds.  Let
$\NF_G(\cdot)$ be the normal form associated to this Gr\"obner basis
(as defined as the \emph{remainder of the division by $G$} in
\cite[Chapter 2, \S 6, Proposition 1]{CoxLitShe97}).  For each
generator $f$ of $\widetilde{\mathcal I_r}$ (i.e. either a maximal
minor of the matrix $\mathcal U$, or a polynomial $u_{i,j}-f_{i,j}$),
we have that $\NF_G(f)=0$.  During the computation of $\NF_G(f)$ by
using the division Algorithm in \cite[Chapter 2, \S 3]{CoxLitShe97}, a
finite set of polynomials (in $\mathbb K(\mathfrak a)[U,X]$) is
constructed. Let $Q_3\in \mathbb K[\mathfrak a]$ denote the product of
the numerators and denominators of all their nonzero
coefficients. Consequently, if $Q_3^{(f)}(\mathbf a)\neq 0$, then $\NF_{\varphi_{\bf a}(G)}(\varphi_{\bf a}(f))=0$ and hence
$\varphi_{\bf a}(f)\in \langle\varphi_{\bf a}(G)\rangle$.  Repeating
this operation for all the generators of $\widetilde{\mathcal I_r}$
yields a finite set of non-identically null polynomials $Q_3^{(f)}\in \mathbb
K[\mathfrak a]$. Let $Q_4\in \mathbb K[\mathfrak a]$ denote their
product. Therefore, if $Q_4(\mathbf a)\neq 0$, then $\varphi_{\bf
  a}(\widetilde{\mathcal I_r})\subset\langle \varphi_{\bf
  a}(G)\rangle$.

Finally, consider the non-empty Zariski open subset $O\subset\mathbb K^{n m \binom{D+k-1}{D}}$ defined by the inequality $Q_1\cdot Q_2\cdot Q_4\neq 0$. For all $\mathbf a\in O$, we have
 $\varphi_{\bf a}(\widetilde{\mathcal
  I_r})=\langle \varphi_{\bf a}(G)\rangle$.

\end{proof}

\begin{cor}\label{coro:LMtransfer}
The leading monomials of $\widetilde{\mathcal I_r}$ are the same as that of $\widetilde{\D_r}$:
$$\LM\left(\widetilde{\mathcal I_r}\right)=\LM\left(\widetilde{\D_r}\right).$$
\end{cor}

\begin{proof}
  By Lemmas \ref{lem:transfer} and \ref{lem:LMgen}, the property
  $\mathcal P_{\widetilde{\mathcal I_r}}$ (resp. $\mathcal
  P_{\widetilde{\D_r}}$) is $\widetilde{\mathcal I_r}$-generic
  and $\widetilde{\D_r}$-generic.
Since $\mathcal P_{\widetilde{\D_r}}$
  (resp. $\mathcal P_{\widetilde{\mathcal I_r}}$) is
  $\widetilde{\D_r}$-generic, there exists a non-empty Zariski
  open subset $O_1\subset \overline{\mathbb
    K}^{nm\left(\binom{D-1+k}{D}+nm\right)}$ (resp. $O_2\subset
  \overline{\mathbb K}^{nm\left(\binom{D-1+k}{D}+nm\right)}$) such
  that, for $(\mathbf b,\mathbf c)\in O_1$ (resp. $O_2$), $\LM\left(\psi_{(\mathbf b, \mathbf c)}(\widetilde{\D_r})\right)=\LM\left(\widetilde{\D_r}\right)$ (resp. $\LM\left(\psi_{(\mathbf b, \mathbf c)}(\widetilde{\D_r})\right)=\LM\left(\widetilde{\mathcal I_r}\right)$).

  Notice that $O_1\cap O_2$ is not empty, since for the Zariski
  topology, the intersection of finitely-many non-empty open subsets
  is non-empty. Let $(\mathbf b, \mathbf c)$ be an
  element of $O_1\cap O_2$. Then
$$\LM\left(\widetilde{\mathcal I_r}\right)=\LM\left(\psi_{(\mathbf b, \mathbf c)}(\widetilde{\D_r})\right)=\LM\left(\widetilde{\D_r}\right).$$
\end{proof}

\begin{cor}\label{coro:HSIrDr}
  The weighted Hilbert series of $\widetilde{\mathcal I_r}$ is the
  same as that of $\widetilde{\D_r}$.
\end{cor}
\begin{proof}
  It is well-known that, for any positively graded ideal $I$ and for any monomial
  ordering, $\wHS_I(t)=\wHS_{\LM(I)}(t)$ (see e.g. the proof of
  \cite[Chapter 9, \S 3, Proposition 9]{CoxLitShe97} which is also valid for quasi-homogeneous ideals). 
  By Corollary \ref{coro:LMtransfer}, $\LM\left(\widetilde{\mathcal I_r}\right)=\LM\left(\widetilde{\D_r}\right)$, which implies that
$$\wHS_{\LM\left(\widetilde{\mathcal I_r}\right)}(t)=\wHS_{\LM\left(\widetilde{\D_r}\right)}(t),$$
and hence $\wHS_{\widetilde{\mathcal I_r}}(t)=\wHS_{\widetilde{\D_r}}(t)$.
\end{proof}

\section{The case $k\geq (n-r)(m-r)$}
\label{sec:welldef}
As we will see in the sequel, the Krull dimension of the ring $\mathbb
K(\mathfrak a)[X]/\mathcal I_r$ is equal to $\max(k- (n-r)(m-r),
0)$. This section is devoted to the study of the case $k\geq
(n-r)(m-r)$.

We show here that the algebraic structure of the ideal
$\mathcal I_r$ is closely related to that of a generic section of a
determinantal variety. 

We recall that the polynomials $g_\ell$ are defined by 
$$g_\ell=\sum_{t\in\mon(D,k)}\mathfrak b_{t}^{(\ell)} t + \sum_{\substack{1\leq i\leq n\\1\leq j\leq m}}\mathfrak c_{i,j}^{(\ell)} u_{i,j}.$$

\begin{lem}\label{lem:Pdim0}
  Let $1\leq \ell\leq nm$ be an integer. If $g_\ell$ divides zero in  
  $\mathbb K(\mathfrak b, \mathfrak c)[U,X]/\left(\D_r+\langle
    g_1,\ldots, g_{\ell-1}\rangle\right)$, then there exists a prime
  ideal $P$ associated to $\D_r+\langle g_1,\ldots,
  g_{\ell-1}\rangle$ such that $\dim(P)=0$.
\end{lem}

\begin{proof}
  If $g_\ell$ divides zero in
  $\mathbb K(\mathfrak b, \mathfrak c)[U,X]/\left(\D_r+\langle
    g_1,\ldots, g_{\ell-1}\rangle\right)$, then there exists a prime
  ideal $P$ associated to $\D_r+\langle g_1,\ldots,
  g_{\ell-1}\rangle$ such that $g_\ell\in P$.
For $\ell\leq nm$, let $\mathfrak b^{(\leq \ell)}$ and $\mathfrak c^{(\leq \ell)}$ denote the sets of parameters

$$\begin{array}{rcl}
\mathfrak b^{(\leq \ell)}&=&\{\mathfrak b^{(s)}_{t} \mid t\in \mon(D,k), 1\leq s\leq \ell\}\\
\mathfrak c^{(\leq \ell)}&=&\{\mathfrak c^{(s)}_{i,j} \mid 1\leq i\leq n, 1\leq j\leq m, 1\leq s\leq \ell\}.
\end{array}$$

Since $\left(\D_r+\langle g_1,\ldots,
  g_{\ell-1}\rangle\right)$ is an ideal of $\mathbb K(\mathfrak
b^{(\leq \ell-1)}, \mathfrak c^{(\leq \ell-1)})[U,X]$, and $P$ is an
associated prime, there exists a Gr\"obner basis $G_P$ of $P$ (for any
monomial ordering $\prec$) which is a finite subset of $\mathbb
K(\mathfrak b^{(\leq \ell-1)}, \mathfrak c^{(\leq \ell-1)})[U,X]$. 

Let $\NF_P(\cdot)$ denote the normal form associated to this Gr\"obner
basis (as defined as the \emph{remainder of the division by $G_P$} in
\cite[Chapter 2, \S 6, Proposition 1]{CoxLitShe97}).

Since $g_\ell\in P$, we have $\NF_P(g_\ell)=0$. By linearity of $\NF_P(\cdot)$, we obtain
$$\sum_{t\in\mon(D, k)}\mathfrak b_{t}^{(\ell)} \NF_P(t) + \sum_{\substack{1\leq i\leq n\\1\leq j\leq m}}\mathfrak c_{i,j}^{(\ell)} \NF_P(u_{i,j}) = 0.$$

Since $G_p\subset\mathbb K(\mathfrak b^{(\leq \ell-1)}, \mathfrak
c^{(\leq \ell-1)})[U,X]$, we can deduce that for any monomial
$t$, $\NF_P(t)\in \mathbb K(\mathfrak b^{(\leq
  \ell-1)}, \mathfrak c^{(\leq \ell-1)})[U,X]$. Therefore, by
algebraic independence of the parameters, the following properties hold: 
for all $t\in \mon(D,k)$, $\NF_P(t)=0$, and for
all $i, j$, $\NF_P(u_{i,j})=0$. Consequently, all monomials of weight degree $D$
in $\mathbb K(\mathfrak b,\mathfrak c)[U,X]$ are in $P$, and hence $P$
has dimension $0$.
\end{proof}

\begin{lem}\label{lem:nonzerodiv}
  For all $\ell\in \{2,\ldots,nm\}$, the polynomial $g_\ell$ does not
  divide zero in $\mathbb K(\mathfrak b,\mathfrak c)[U,X]/(\D_r+\langle g_1,\ldots, g_{\ell-1}\rangle)$ and $\dim(\D_r+\langle g_1,\ldots, g_{\ell}\rangle) = k +( n+m-r)r - \ell$.
\end{lem}

\begin{proof}
  We prove the Lemma by induction on $\ell$. According to
  \cite[Corollary 2 of Theorem 1]{HocEag70}, the ring $\mathbb
  K(\mathfrak b,\mathfrak c)[U,X]/\D_r$ is Cohen-Macaulay and
  purely equidimensional. First, notice that the dimension is equal to
  $k +( n+m-r)r$ for $\ell=0$ since the dimension of the ideal
  $\D_r\subset\mathbb K[U]$ is $(n+m-r)r$ (see e.g. \cite{ConHer94} and references therein). Now, suppose that the dimension of the ideal $\D_r+\langle g_1,\ldots,g_{\ell-1}\rangle\subset \mathbb K(\mathfrak
  b, \mathfrak c)[U,X]$ is $k +( n+m-r)r - \ell+1$. Since the ring
  $\mathbb K(\mathfrak b,\mathfrak c)[U,X]/\D_r$ is
  Cohen-Macaulay and $\langle g_1,\ldots,g_{\ell-1}\rangle$ has
  co-dimension $\ell-1$ in $\mathbb K(\mathfrak b, \mathfrak c)[U,X]$,
  the Macaulay unmixedness Theorem \cite[Corollary
  18.14]{Eis95} implies that $\langle g_1,\ldots, g_{\ell-1} \rangle$ has no
  embedded component and is equidimensional in $\mathbb K(\mathfrak
  b,\mathfrak c)[U,X]/\D_r$. Hence $\D_r+\langle
  g_1,\ldots, g_{\ell-1} \rangle$ as an ideal in $\mathbb K(\mathfrak
  b,\mathfrak c)[U,X]$ has no embedded component and is
  equidimensional.  By contradiction, suppose that $g_\ell$ divides
  zero in $\mathbb K(\mathfrak b,\mathfrak c)[U,X]/(\D_r+\langle g_1,\ldots, g_{\ell-1}\rangle)$. By Lemma
  \ref{lem:Pdim0}, there exists a prime $P$ associated to $\D_r+\langle g_1,\ldots,g_{\ell-1}\rangle$ such that $\dim(P)=0$,
  which contradicts the fact that $\D_r+\langle g_1,\ldots,
  g_{\ell-1}\rangle$ is purely equidimensional of dimension $k +(
  n+m-r)r - \ell+1>0$.
\end{proof}

\begin{lem} \label{lem:HSeq}
  The Hilbert series of the $\mathcal I_r\subset\mathbb K(\mathfrak
  a)[X]$ equals the weighted Hilbert series of $\widetilde{\mathcal I_r}\subset \mathbb
  K(\mathfrak a)[X,U]$.
\end{lem}

\begin{proof}
  Let $\prec_{lex}$ denote a lexicographical ordering on $\mathbb
  K(\mathfrak a)[X,U]$ such that $x_k\prec_{lex} u_{i,j}$ for all
  $k,i,j$. By \cite[Section 9.3, Proposition 9]{CoxLitShe97},
  $\HS_{\mathcal I_r}(t)=\HS_{\LM_{\prec_{lex}}(\mathcal I_r)}(t)$ and
  $\wHS_{\widetilde{\mathcal
      I_r}}(t)=\wHS_{\LM_{\prec_{lex}}(\widetilde{\mathcal I_r})}(t)$.
  Since $\LM_{\prec_{lex}}(u_{i,j}-f_{i,j})=u_{i,j}$, we deduce that all monomials which are multiples of a variable $u_{i,j}$ are in $\LM_{\prec_{lex}}(\widetilde{\mathcal I_r})$. Therefore, the remaining monomials in $\LM_{\prec_{lex}}(\widetilde{\mathcal I_r})$ are in $\mathbb K(\mathfrak a)[X]$:
  $$\begin{array}{rcl}\LM_{\prec_{lex}}(\widetilde{\mathcal I_r})&=&\left\langle \{u_{i,j}\}\cup \LM_{\prec_{lex}}(\widetilde{\mathcal I_r}\cap \mathbb K(\mathfrak a)[X])\right\rangle\\
    &=&\left\langle \{u_{i,j}\}\cup \LM_{\prec_{lex}}(\mathcal I_r)\right\rangle.\end{array}$$
  Therefore, $\frac{\mathbb K(\mathfrak a)[U,X]}{\LM_{\prec_{lex}}(\widetilde{\mathcal I_r})}$ is isomorphic (as a graded $\mathbb K(\mathfrak a)$-algebra) to $\frac{\mathbb K(\mathfrak a)[X]}{\LM_{\prec_{lex}}(\mathcal I_r)}$. Thus $$\HS_{\LM_{\prec_{lex}}(\mathcal I_r)}(t)=\wHS_{\LM_{\prec_{lex}}(\widetilde{\mathcal I_r})}(t),$$ and hence
  $$\HS_{\mathcal I_r}(t)=\wHS_{\widetilde{\mathcal I_r}}(t).$$
\end{proof}

In the sequel, $A_r(t)$ denotes the $r\times r$ matrix whose
$(i,j)$-entry is $\sum_k \binom{m-i}{k} \binom{n-j}{k} t^k$.
The following theorem is the main result of this section:

\begin{thm}\label{theo:HSIr}
The dimension of the ideal $\mathcal I_r$ is $k-(n-r)(m-r)$ and its Hilbert series is
$$\HS_{\mathcal I_r}(t)=\frac{\det \left(A_r(t^D)\right) (1-t^D)^{(n-r)(m-r)}}{t^{D\binom{r}{2}} (1-t)^{k}}.$$
\end{thm}

\begin{proof}
  According to \cite[Corollary 1]{ConHer94} (and references
  therein), the ideal $\D_r$ seen as an ideal of $\mathbb
  K[U]$ has dimension $(m+n-r)r$ and its Hilbert series (for the
  standard gradation: $\deg(u_{i,j})=1$) is the power series expansion of
$$\HS_{\D_r\subset \mathbb K[U]}(t)=\frac{\det A_r(t)}{t^{\binom{r}{2}} (1-t)^{(n+m-r)r}}.$$
By putting a weight $D$ on each variable $u_{i,j}$ (i.e. $\deg(u_{i,j})=D$), the weighted
Hilbert series of $\D_r\subset \mathbb K[U]$ is
$$\wHS_{\D_r\subset \mathbb K[U]}(t)=\frac{\det A_r(t^D)}{t^{D\binom{r}{2}} (1-t^D)^{(n+m-r)r}}.$$
By considering $\D_r$ as an ideal of $\mathbb K(\mathfrak
b,\mathfrak c)[U,X]$, the dimension becomes $k+(m+n-r)r$ and
its weighted Hilbert series is
$$\wHS_{\D_r\subset \mathbb K(\mathfrak b,\mathfrak c)[U,X]}(t)=\frac{\det A_r(t^D)}{t^{D\binom{r}{2}} (1-t)^{k}(1-t^D)^{(n+m-r)r}}.$$ 

According to Lemma \ref{lem:nonzerodiv}, for each $\ell\leq nm$, the polynomial $g_\ell$ does not divide zero in the ring $$\mathbb K(\mathfrak b, \mathfrak c)[U,X]/(\D_r+\langle g_1,\ldots, g_{\ell-1}\rangle).$$ This implies the following relations:
$$\begin{array}{rcl}\dim\left(\D_r+\langle g_1,\ldots, g_{\ell}\rangle\right)&=&\dim\left(\D_r+\langle g_1,\ldots, g_{\ell-1}\rangle\right)-1\\
  \wHS_{\D_r+\langle g_1,\ldots, g_{\ell}\rangle}(t)&=&(1-t^D)\wHS_{\D_r+\langle g_1,\ldots, g_{\ell-1}\rangle}(t).\end{array}$$

Therefore the dimension of $\widetilde{\D_r}$ is $k-n m+(n+m-r)r$ and its quasi-homogeneous Hilbert series
is
$$\wHS_{\widetilde{\D_r}}(t)=\frac{\det \left(A_r(t^D)\right)}{t^{D\binom{r}{2}}
  (1-t)^{k}(1-t^D)^{(n+m-r)r-n m}}=\frac{\det \left(A_r(t^D)\right) (1-t^D)^{(n-r)(m-r)}}{t^{D\binom{r}{2}} (1-t)^{k}}.$$ By Corollary
\ref{coro:HSIrDr}, the ideal $\widetilde{\mathcal I_r}$ has the same
weighted Hilbert series. Finally, by Lemma \ref{lem:HSeq}, the Hilbert
series of $\mathcal I_r$ $=\widetilde{\mathcal I_r}\cap\mathbb
K(\mathfrak a)[X]$ is the same as that of
$\widetilde{\mathcal I_r}$.
\end{proof}

\begin{cor} \label{coro:degree}
The degree of the ideal $\mathcal I_r$ is:
$$\begin{array}{rcl}\DEG(\mathcal I_r)&=&\displaystyle D^{(n-r)(m-r)} \prod_{i=0}^{m-r-1}\frac{i! (n+i)!}{(m-1-i)! (n-r+i)!}\\
&=&\displaystyle D^{(n-r)(m-r)} \prod_{i=0}^{m-r-1}\frac{\binom{n+m-r-1}{r+i}}{\binom{n+m-r-1}{i}}.\end{array}$$
\end{cor}

\begin{proof}
From \cite[Example 14.4.14]{Ful97}, the degree of the ideal $\D_r$ is 
$$\prod_{i=0}^{m-r-1}\frac{i! (n+i)!}{(m-1-i)! (n-r+i)!}.$$
Since the degree is equal to the numerator of the Hilbert series of $\D_r$ evaluated at $t=1$, 
$$\det A_r(1)=\prod_{i=0}^{m-r-1}\frac{i! (n+i)!}{(m-1-i)! (n-r+i)!}.$$
By Theorem \ref{theo:HSIr}, the Hilbert series of $\mathcal I_r$ is
$$\begin{array}{rcl}
\HS_{\mathcal I_r}(t)&=&\displaystyle\frac{\det \left(A_r(t^D)\right)(1-t^D)^{(n-r)(m-r)} }{t^{D\binom{r}{2}} (1-t)^{k}}\\
&=&\displaystyle\frac{\det \left(A_r(t^D)\right) (1+t+\dots+t^{D-1})^{(n-r)(m-r)}}{t^{D\binom{r}{2}}(1-t)^{k-(n-r)(m-r)}}.
\end{array}
$$
Thus, the evaluation of the numerator in $t=1$ yields
$$\DEG(\mathcal I_r)=D^{(n-r)(m-r)} \prod_{i=0}^{m-r-1}\frac{i! (n+i)!}{(m-1-i)! (n-r+i)!}.$$
To prove the second equality, notice that
$$\prod_{i=0}^{m-r-1}\frac{\binom{n+m-r-1}{r+i}}{\binom{n+m-r-1}{i}}=\prod_{i=0}^{m-r-1}\frac{i!(n+m-r-i-1)!}{(r+i)!(n+m-2r-i-1)!}.$$
By substituting $i$ by $m-r-1-i$, we obtain that
$$\begin{array}{rcl}
\displaystyle\prod_{i=0}^{m-r-1} (n+m-r-i-1)!&=&\displaystyle\prod_{i=0}^{m-r-1} (n+i)!\\
\displaystyle\prod_{i=0}^{m-r-1}(r+i)! &=&\displaystyle\prod_{i=0}^{m-r-1} (m-i-1)!\\
\displaystyle\prod_{i=0}^{m-r-1} (n+m-2r-i-1)!&=&\displaystyle\prod_{i=0}^{m-r-1} (n-r+i)!.
\end{array}$$
Consequently,
$$\displaystyle\prod_{i=0}^{m-r-1}\frac{i! (n+i)!}{(m-1-i)! (n-r+i)!}=\prod_{i=0}^{m-r-1}\frac{\binom{n+m-r-1}{r+i}}{\binom{n+m-r-1}{i}}.$$
\end{proof}

\section{The over-determined case}
\label{sec:overdef}

To study the over-determined case ($k<(n-r)(m-r)$), we need to assume a variant of Fr\"oberg's conjecture \cite{Fro85}:
\begin{conj}\label{conj:froberg}
 Let
$\D_{\ell,i}$ denote the vector space of quasi-homogeneous
polynomials of weight degree $i$ in $\D_r+\langle g_1,\ldots, g_\ell\rangle$. 
Then the linear map
$$\begin{array}{ccc}
\mathbb K(\mathfrak b,\mathfrak c)[U,X]_i/\D_{\ell,i}&\longrightarrow&\mathbb K(\mathfrak b,\mathfrak c)[U,X]_{i+D}/\D_{\ell,i+D}\\
f&\longmapsto&f g_{\ell+1}
\end{array}$$
has maximal rank, i.e. it is either injective or onto.
\end{conj}

\begin{rem}
If $k+(n+m-r)r-\ell>0$, then Conjecture \label{conj:froberg} is proved by Lemma \ref{lem:nonzerodiv}: $g_{\ell+1}$ does not divide zero in $\mathbb K(\mathfrak b,\mathfrak c)[U,X]/\left(\D_r+\langle g_1,\ldots, g_\ell\rangle\right)$ and hence the linear map is injective for all $i\in\mathbb N$.
\end{rem}

{\bf Notation.} Given a power series $S(t)\in \mathbb Z[[t]]$, we
let $[S(t)]_+$ denote the power series obtained by truncated $S(t)$ at
its first non positive coefficient. 

\begin{lem}\label{lem:froberg}
If Conjecture \ref{conj:froberg} is true, then the Hilbert series of
$\D_r+\langle g_1,\ldots, g_{\ell+1}\rangle$ is
$$\wHS_{\D_r+\langle g_1,\ldots, g_{\ell+1}\rangle}(t)=\left[(1-t^D)\wHS_{\D_r+\langle g_1,\ldots, g_\ell\rangle}(t)\right]_+.$$
\end{lem}

\begin{proof}
In this proof, for
simplicity of notation, we let $R$ denote the ring $\mathbb
K(\mathfrak b,\mathfrak c)[U,X]$. 
If $S(t)=\sum_{i\in \mathbb N}s_i t^i\in \mathbb Z[[t]]$ is a power series, $\left[S(t)\right]_{\geq 0}$ denotes the series 
$$\left[S(t)\right]_{\geq 0}=\sum_{i\in \mathbb N} \max (s_i,0) t^i.$$
 Let $\mathsf{ann}(g_{\ell+1})$ be
the ideal $\{f\in R : f g_{\ell+1}\in \D_r+\langle g_1,\ldots,
g_\ell\rangle\}$. For $i\in \mathbb N$, consider the following exact
sequence:
$$\begin{array}{r}0\rightarrow \mathsf{ann}(g_{\ell+1})_i\rightarrow
  R_{i}/\D_{\ell,i}\xrightarrow{\times g_{\ell+1}}
  R_{i+D}/\D_{\ell,i+D}\rightarrow \\
\rightarrow R_{i+D}/\D_{\ell+1,i+D}\rightarrow 0.
\end{array}$$
By Conjecture \ref{conj:froberg}, we obtain
$$\dim(\mathsf{ann}(g_{\ell+1})_i)=\max(0,\dim(R_i/\D_{\ell,i})-\dim(R_{i+D}/\D_{\ell,i+D})).$$ The alternate sum of the dimensions of the
vector spaces occurring in an exact sequence is zero; it follows that
$$\begin{array}{rcl}
  \dim( R_{i+D}/\D_{\ell+1,i+D})&=&\dim(R_{i+D}/\D_{\ell,i+D})-\dim(R_{i}/\D_{\ell,i})+\\&&\max(0,\dim(R_i/\D_{\ell,i})-\dim(R_{i+D}/\D_{\ell,i+D}))\\
  &=&\max(0,\dim(R_{i+D}/\D_{\ell,i+D})-\dim(R_{i}/\D_{\ell,i})).
\end{array}$$
Multiplying this identity by $t^{i+D}$ yields
$$\begin{array}{rcl}
  \left[t^{i+D}\right]\wHS_{\D_r+\langle g_1,\ldots,g_{\ell+1}\rangle}(t)&=&\dim\left(R_{i+D}/\D_{\ell+1,i+D)}\right)\\
  &=&\max\left(0,\dim(R_{i+D}/\D_{\ell,i+D})-\dim(R_{i}/\D_{\ell,i})\right)\\
  &=& \max\left(0,[t^{i+D}](1-t^D) \wHS_{\D_r+\langle g_1,\ldots,g_{\ell}\rangle}(t)\right)\\
  &=&[t^{i+D}]\left[(1-t^D)\wHS_{\D_r+\langle g_1,\ldots,g_{\ell}\rangle}(t)\right]_{\geq 0}.
\end{array}$$
Since any monomial in $\mathbb K(\mathfrak a)[X,U]$ of weight degree greater that $D$ is a multiple of a monomial of weight degree $D$, we deduce that if there exists $i_0\geq D$ such that $$\left[t^{i_0}\right]\wHS_{\D_r+\langle g_1,\ldots,g_{\ell+1}\rangle}(t) = 0,$$ then for all $i>i_0$, $\left[t^{i}\right]\wHS_{\D_r+\langle g_1,\ldots,g_{\ell+1}\rangle}(t) = 0$.
Therefore $$\left[t^{i+D}\right]\wHS_{\D_r+\langle g_1,\ldots,g_{\ell+1}\rangle}(t)=[t^{i+D}]\left[(1-t^D)\wHS_{\D_r+\langle g_1,\ldots,g_{\ell}\rangle}(t)\right]_{+},$$
Finally, by summing over $i$, we get
$$\wHS_{\D_r+\langle g_1,\ldots,g_{\ell+1}\rangle}(t)=\left[(1-t^D)\mathsf{HS}_{\D_r+\langle g_1,\ldots, g_\ell\rangle}(t)\right]_+.$$
\end{proof}

\begin{thm}
If Conjecture \ref{conj:froberg} is true, then the
Hilbert series of $\mathcal I_r$ is
$$\mathsf{HS}_{\mathcal I_r}(t)=\left[(1-t^D)^{(n-r)(m-r)} \frac{\det \left(A_r(t^D)\right)}{t^{D\binom{r}{2}}(1-t)^k}\right]_+,$$
where $A_r(t)$ is the $r\times r$ matrix whose $(i,j)$-entry is $\displaystyle\sum_{k=0}^{\min(m-i,n-j)} \binom{m-i}{k}\binom{n-j}{k} t^k$.
\end{thm}

\begin{proof}
By applying  $n m$ times Lemma \ref{lem:froberg}, we obtain that 
$$\mathsf{wHS}_{\widetilde{\D_r}}(t)=\left[(1-t^D)\left[(1-t^D)\ldots \left[(1-t^D) \frac{\det A_r(t^D)}{t^{D\binom{r}{2}}(1-t)^{k}(1-t^D)^{(n+m-r)r}}\right]_+\ldots\right]_+\right]_+.$$
Let $S=\sum_{0\leq i} a_i t^i \in\mathbb Z[[t]]$ be a power series such that $a_0>0$, and let $i_0\in\mathbb N\cup\{\infty\}$ be defined as 
$$i_0=\begin{cases}\infty\text{ if for all }i\geq 0, a_i>0;\\
 \min(\{i \mid a_i\leq 0\})\text{ otherwise.}\end{cases}$$ Therefore, $\left[S(t)\right]_+ =\sum_{0\leq i<i_0} a_i t^i$. By convention, for $i<0$, we put $a_i=0$. Then 
$$\begin{array}{rcl}
(1-t^D)S(t)&=&\sum_{0\leq i} (a_i-a_{i-D}) t^i\\
(1-t^D)\left[S(t)\right]_+&=&\sum_{0\leq i< i_0} (a_i-a_{i-D}) t^i
\end{array}.$$
Consequently, the coefficients of $(1-t^D)S(t)$ and of $(1-t^D)\left[S(t)\right]_+$ are equal up to the index $i_0$.

\begin{itemize}
\item If $i_0=\infty$, then $(1-t^D)S(t)=(1-t^D)\left[S(t)\right]_+$ and hence $$\left[(1-t^D)S(t)\right]_+=\left[(1-t^D)\left[S(t)\right]_+\right]_+;$$
\item if $i_0<\infty$, then $a_{i_0-D}$ is positive and thus $a_{i_0}-a_{i_0-D}$ is negative. Let $i_1$ be the index of the first non-positive coefficient of $(1-t^D)S(t)$. Then $i_1<i_0$, and hence $\left[(1-t^D)S(t)\right]_+=\left[(1-t^D)\left[S(t)\right]_+\right]_+$.
\end{itemize}

Therefore, for all power series $S\in \mathbb Z[[t]]$ such that $S(0)>0$, we have
$$\left[(1-t^D)\left[S\right]_+\right]_+ = \left[(1-t^D) S\right]_+.$$
Consequently, an induction shows that 
$$\mathsf{wHS}_{\widetilde{\D_r}}(t)=\left[(1-t^D)^{(n-r)(m-r)} \frac{\det A(t^D)}{t^{D\binom{r}{2}}(1-t)^k}\right]_+.$$

Then, by Corollary \ref{coro:HSIrDr}, $\mathsf{wHS}_{\widetilde{\D_r}}(t) = \mathsf{wHS}_{\widetilde{\mathcal I_r}}(t)$. Finally, by Lemma \ref{lem:HSeq}, we conclude that $\HS_{\mathcal I_r}(t)=\wHS_{\widetilde{\mathcal I_r}}(t).$
\end{proof}

\section{Complexity analysis}
\label{sec:compl}
Using the previous results on the Hilbert series of $\mathcal I_r$, we
analyze now the arithmetic complexity of solving the generalized
MinRank problem with Gr\"obner bases algorithms. In the first part of
this section (until Section \ref{sec:affinecompl}), we consider the
homogeneous MinRank problem (i.e. the polynomials $f_{i,j}$ are
homogeneous).

Computing a Gr\"obner basis of the ideal $\varphi_{\bf a}(\mathcal
I_r)$ for the lexicographical ordering yields an explicit description
of the set of points $V$ such that the matrix
$$\varphi_{\bf a}(\mathcal M)= \begin{pmatrix}
\varphi_{\bf a}(f_{1,1})&\dots&\varphi_{\bf a}(f_{1,m})\\ \vdots&\ddots&\vdots\\ \varphi_{\bf a}(f_{n,1})&\dots&\varphi_{\bf a}(f_{n,m})
\end{pmatrix}$$ has rank less than $r+1$.  In this section, we study
the complexity of this computation when $\mathbf a\in \mathbb K^{n m\binom{k+D-1}{D}}$
is generic (i.e. $\mathbf a$ belongs to a given non-empty Zariski open subset of $\overline{\mathbb K}^{n m\binom{k+D-1}{D}}$) by using the theoretical results from Sections
\ref{sec:welldef} and \ref{sec:overdef}.  We focus on the
$0$-dimensional cases $k = (n-r)(m-r)$ and $k
< (n-r)(m-r)$ (over-determined case). Therefore, the set of points where
the evaluation of the matrix $\varphi_{\bf a}(\mathcal M)$ has rank less than $r+1$ is
finite.

In order to compute this set of points, we use the following strategy:
\begin{itemize}
\item compute a Gr\"obner basis of $\varphi_{\bf a}(\mathcal I_r)$ for the \emph{grevlex} (graded reverse lexicographical) ordering with the $F_5$ algorithm \cite{Fau02};
\item convert it into a lexicographical Gr\"obner basis of $\varphi_{\bf a}(\mathcal I_r$) by using the FGLM algorithm \cite{FauGiaLazMor93,FauMou11}.
\end{itemize}

First, we recall some results about the complexity of the algorithms
$F_5$ and FGLM. The two quantities which allow us to estimate their
complexity are respectively the \emph{degree of regularity} and the
\emph{degree} of the ideal. The degree of regularity of a
$0$-dimensional homogeneous ideal $I$ is the smallest integer $d$ such
that all monomials of degree $d$ are in $I$; it is independent on the
monomial ordering and it bounds the degrees of the polynomials in a
minimal Gr\"obner basis of $I$. Moreover, in the $0$-dimensional case,
the Hilbert series is a polynomial from which the degree of regularity
can be read off: $\dreg(I) = \deg(\HS_I(t))+1$.

In the sequel, $\omega$ denotes a feasible exponent for the matrix
multiplication (i.e. a number such that there exists an deterministic
algorithm which computes the product of two $n\times n$ matrices in
$O\left(n^\omega\right)$ arithmetic operations in $\mathbb K$). The
best known bound on this exponent is $\omega<2.3727$ \cite{Wil11}.

The following proposition and its proof are a variant of a result
known in the context of semi-regular sequences (see e.g. \cite{Laz83} and \cite{Fau99} for the relation between Gr\"obner
basis computation and linear algebra, \cite[Proposition 10]{BarFauSal05} and \cite[Section 3.4]{Bar04} for the complexity
analysis).

\begin{prop}[\cite{BarFauSal05,Bar04}]\label{prop:complF5}
Let $h_1,\ldots, h_\ell\in \mathbb K[x_1,\ldots, x_k]$ be homogeneous polynomials of degrees $d_1,\ldots, d_\ell$, and $I=\langle h_1,\ldots, h_\ell\rangle$.
The complexity of computing a Gr\"obner basis of $I$ for a monomial ordering $\prec$ is upper bounded by 
$$O\left(\left(\binom{k+\dreg(I)}{\dreg(I)}-\DEG(I)\right)^{\omega-2}\binom{k+\dreg(I)}{\dreg(I)}\sum_{i=1}^{\ell} \binom{k+\dreg(I)-d_i}{\dreg(I)-d_i}\right).$$
\end{prop}

\begin{proof}
Since $I$ is homogeneous, a Gr\"obner basis can be obtained by computing the row echelon form of the so-called
\emph{Macaulay matrix} of the system up to degree $\dreg(I)$. This matrix is constructed as follows:
\begin{itemize}
\item the rows are indexed by the products $t h_i$, where $1\leq i\leq \ell$ and $t\in \mathbb K[x_1,\ldots, x_k]$ is a monomial of degree at most $\dreg(I)-d_i$;
\item the columns are indexed by the monomials $m\in \mathbb K[x_1,\ldots, x_k]$ of degree at most $\dreg(I)$ and are sorted in descending order with respect to $\prec$;
\item the coefficient at the intersection of the row $t h_i$ and the column $m$ is the coefficient of $m$ in the polynomial $t h_i$.
\end{itemize}

The number of columns of this matrix is the number of monomials in
$\mathbb K[x_1,\ldots,x_k]$ of degree at most $\dreg(I)$, namely
$\binom{k+\dreg(I)}{\dreg(I)}$. The number of rows is 
$\sum_{i=1}^{\ell} \binom{k+\dreg(I)-d_i}{\dreg(I)-d_i}$ and its rank is
equal to $\left(\binom{k+\dreg(I)}{\dreg(I)}-\DEG(I)\right)$. 

According to \cite[Theorem 2.10]{Sto00}, the complexity of computing the row echelon form of a $p\times q$ matrix of rank $r$ is upper bounded by $O(r^{\omega-2}p q)$.

Consequently,
the complexity of computing a Gr\"obner basis of $I$ is upper bounded by 
$$O\left(\left(\binom{k+\dreg(I)}{\dreg(I)}-\DEG(I)\right)^{\omega-2}\binom{k+\dreg(I)}{\dreg(I)}\sum_{i=1}^{\ell} \binom{k+\dreg(I)-d_i}{\dreg(I)-d_i}\right).$$
\end{proof}

\begin{rem}
Notice that 
$$\begin{array}{rcl}
\displaystyle \binom{k+\dreg(I)}{\dreg(I)}-\DEG(I)&\leq& \displaystyle \binom{k+\dreg(I)}{\dreg(I)}\\
\displaystyle \sum_{i=1}^{\ell} \binom{k+\dreg(I)-d_i}{\dreg(I)-d_i}&\leq&\displaystyle  \ell \binom{k+\dreg(I)}{\dreg(I)}.
\end{array}$$
Therefore, the complexity of computing a Gr\"obner basis of $I$ can also be upper bounded by the simpler expression $O\left(\ell \binom{k+\dreg(I)}{\dreg(I)}^\omega\right)$.
\end{rem}

\begin{lem}\label{lem:dregCompl}
If $k = (n-r)(m-r)$, then the degree of regularity of $\mathcal I_r$ is
$$\dreg\left(\mathcal I_r\right)=D r(m-r) +(D-1)k+1.$$
\end{lem}

\begin{proof}
According to Theorem \ref{theo:HSIr}, the Hilbert series of $\mathcal I_r$ is
$$\HS_{\mathcal I_r}(t)=\frac{\det A_r(t^D) (1-t^D)^{(n-r)(m-r)}}{t^{D\binom{r}{2}} (1-t)^{k}}.$$
By definition of the matrix $A_r(t)$, the highest degree on each row is reached on the diagonal. Thus, the degree of $\det(A_r(t))$ is the degree of the product of its diagonal elements:
$$\deg(\det(A_r(t))) = \sum_{i=1}^r (\min(n,m)-i) = r m - \binom{r+1}{2}.$$
Therefore, we can compute the degree of the Hilbert series which is a
polynomial since the ideal is $0$-dimensional:
$$\begin{array}{rcl}
\dreg\left(\mathcal I_r\right)&=&\deg(\HS_{\mathcal I_r}(t))+1\\
&=&\deg(\det(A_r(t^D)))+D \left(n-r\right)\left(m-r\right) - D \binom{r}{2} - k+1\\
&=&D(r m - \binom{r+1}{2} +n m -(n+m-r)r - \binom{r}{2})-k+1\\
&=&D r(m-r) +(D-1)k+1.
\end{array}$$
\end{proof}

\begin{cor}\label{coro:dregCompl}
  If $k = (n-r)(m-r)$, then there exists a non-empty Zariski open
  subset $O\subset \overline{\mathbb K}^{n m \binom{D-1+k}{D}}$ such
  that for all $\mathbf a\in O$, the degree of regularity of
  $\varphi_{\bf a} (\mathcal I_r)$ is
$$\dreg\left(\varphi_{\bf a} (\mathcal I_r)\right)=D r(m-r) +(D-1)k+1.$$
\end{cor}

\begin{proof}
  According to Lemma \ref{lem:LMgen}, there exists a Zariski open subset $O$ such
  that for all $\mathbf a\in O$, $\LM\left(\mathcal
    I_r\right)=\LM\left(\varphi_{\mathbf a}(\mathcal
    I_r)\right)$. Consequently, the polynomials in minimal Gr\"obner
  bases of $\mathcal I_r$ and $\varphi_{\bf a}(\mathcal I_r)$ have the
  same leading monomials. Since the degree of regularity is the
  highest degree of the polynomials in a minimal Gr\"obner basis, we have
  $\dreg\left(\varphi_{\bf a} (\mathcal
    I_r)\right)=\dreg\left(\mathcal I_r\right)$. Lemma
  \ref{lem:dregCompl} concludes the proof.
\end{proof}

The degree of regularity governs the complexity of the Gr\"obner basis
computation with respect to the grevlex ordering. The complexity of
the algorithm FGLM is upper bounded by $O(k\cdot \DEG(I)^3)$ which is
polynomial in the degree of the ideal \cite{FauGiaLazMor93,FauMou11}.

We can now state the main complexity result:

\begin{thm}\label{theo:complHom}
 There exists a non-empty Zariski open subset $O\subset\overline{\mathbb K}^{n m
  \binom{D-1+k}{D}}$ such that for any $\mathbf a\in O$, the
arithmetic complexity of computing a lexicographical Gr\"obner basis of the ideal generated by the $(r+1)\times(r+1)$-minors of the matrix
$\varphi_{\mathbf a}(\mathcal M)$ is upper bounded by
$$O\left(\binom{n}{r+1}\binom{m}{r+1}\binom{\dreg(\varphi_{\mathbf a}(\mathcal I_r)+k}{k}^\omega + k\left(\DEG\left(\varphi_{\mathbf a}(\mathcal I_r)\right)\right)^3\right),$$
where $2\leq \omega\leq 3$ is a feasible exponent for the matrix multiplication, and 

\begin{itemize}
\item if $k = (n-r)(m-r)$, then 
$$\dreg(\varphi_{\mathbf a}(\mathcal I_r)=\deg(\HS_{\varphi_{\mathbf a}(\mathcal I_r)}(t))+1=D r(m-r) +(D-1)k+1$$
 and $\DEG(\varphi_{\mathbf a}(\mathcal I_r))=\HS_{\varphi_{\mathbf a}(\mathcal I_r)}(1)=D^{n m-(n+m-r) r} \prod_{i=0}^{m-r-1}\frac{i! (n+i)!}{(m-1-i)! (n-r+i)!}$.
\item if $k < (n-r)(m-r)$, then assuming that Conjecture \ref{conj:froberg} is true, $$\dreg(\varphi_{\mathbf a}(\mathcal I_r)=\deg(\HS_{\varphi_{\mathbf a}(\mathcal I_r)}(t))+1$$ and $\DEG(\varphi_{\mathbf a}(\mathcal I_r))=\HS_{\varphi_{\mathbf a}(\mathcal I_r)}(1)$ where 
$$\mathsf{HS}_{\varphi_{\mathbf a}(\mathcal I_r)}(t)=\left[(1-t^D)^{n m-(n+m-r)r} \frac{\det A(t^D)}{t^{D\binom{r}{2}}(1-t)^k}\right]_+.$$
\end{itemize}
\end{thm}

\begin{proof}
  The number of $(r+1)$-minors of the matrix $\varphi_a(\mathcal M)$
  is $\binom{n}{r+1}\binom{m}{r+1}$.  Consequently, the theorem is a
  straightforward consequence of the bounds on the complexity of the
  $F_5$ algorithm (Proposition \ref{prop:complF5}) and of the FGLM
  algorithm \cite{FauGiaLazMor93,FauMou11}, together with the formulas
  for the degree of regularity (Corollary \ref{coro:dregCompl}) and
  for the degree (Corollary \ref{coro:degree}).
\end{proof}

\begin{rem}
 There exists a polynomial $h(\mathfrak a)$ in $\mathbb
  Z[\mathfrak a]$ when the characteristic of $\mathbb K$ is $0$, such that
$$h(\mathbf a)\neq 0\Rightarrow \mathbf a\in O.$$ Also
note that this polynomial does not depend on the field $\mathbb K$: if
$\mathbb K=\mathbb F_{q}$ is a finite field ($q=p^e$), then the
polynomial $\bar h (\mathfrak a)$ (where all coefficients are taken
modulo $p$) verifies the requested property. Schwartz-Zippel's Lemma
states that, if $\bf a$ is chosen uniformly at random in $\mathbb
F_q^{n m \binom{D-1+k}{D}}$, the probability that $h(\mathbf a)=0$ is
upper bounded by $\deg(h)/q$ and therefore tends towards $0$ when the
cardinality $q$ of the field tends to infinity. This explains why
these complexity results can be used for practical applications
when $\mathsf{char}(\mathbb K)=0$ or $\mathbb K$ is a sufficiently
large finite field.
\end{rem}

\subsection{Positive dimension}

When $k > (n-r)(m-r)$, the ideal $\mathcal I_r$ has positive dimension. To achieve complexity bounds in that case, we need upper bounds on the maximal degree in a minimal Gr\"obner basis of $\mathcal I_r$.

\begin{lem}\label{lem:degposdim}
  If $k > (n-r)(m-r)$, then the maximal degree in a minimal Gr\"obner basis of $\mathcal I_r$ is bounded by
$$D r(m-r) +(D-1)(n-r)(m-r)+1.$$
\end{lem}

\begin{proof}
  Consider the ideal $J$ obtained by specializing the last $k -
  (n-r)(m-r)$ variables to zero in $\mathcal I_r$. We prove now that
  $\LM(\mathcal I_r)=\LM(J)$. First, notice that for the grevlex
  ordering, $\LM(J)\subset \LM(\mathcal I_r)$. According to
  Theorem \ref{theo:HSIr}, the Hilbert series of the ideal $J\cap \mathbb
  K(\mathfrak a)[x_1,\ldots,x_{(n-r)(m-r)}]$ is equal to
$$\frac{\det A_r(t^D) (1-t^D)^{(n-r)(m-r)}}{t^{D\binom{r}{2}} (1-t)^{(n-r)(m-r)}}.$$
By construction, $J\subset \mathbb K(\mathfrak a)[x_1,\ldots,x_{(n-r)(m-r)}]$, thus the Hilbert series of $J$ as an ideal of the ring $\mathbb
K(\mathfrak a)[x_1,\ldots,x_k]$ is equal to
$$\frac{\det A_r(t^D) (1-t^D)^{(n-r)(m-r)}}{t^{D\binom{r}{2}} (1-t)^{k}},$$
which is equal to the Hilbert series of $\mathcal I_r$.

Since $\HS_J(t)=\HS_{\mathcal I_r}(t)$ and $\LM(J)\subset \LM(\mathcal I_r)$, we can deduce that  $\LM(J) = \LM(\mathcal I_r)$.

Consequently, the leading monomials in minimal Gr\"obner
bases of $J$ and $\mathcal I_r$ are the same. Hence, the polynomials
in both Gr\"obner bases have the same degrees since they are
homogeneous.

Finally, notice that the Gr\"obner basis of the ideal  $J$ is the same as that of the ideal
$J\cap \mathbb K(\mathfrak a)[x_1,\ldots,x_{(n-r)(m-r)}]$ which, by Lemma
\ref{lem:dregCompl}, is
a zero-dimensional ideal whose degree of regularity is $D r(m-r)
+(D-1)(n-r)(m-r)+1$. Therefore the maximal degree of the polynomials in the minimal
reduced Gr\"obner basis of $\mathcal I_r$ is bounded by $D r(m-r) +(D-1)(n-r)(m-r)+1$.
\end{proof}

Using exactly the same argumentation as in the proof of Corollary \ref{coro:dregCompl}, we deduce that
\begin{cor}\label{coro:degposdim}
  If $k > (n-r)(m-r)$, then there exists a non-empty Zariski open subset $O\subset \overline{\mathbb K}^{n m
  \binom{D-1+k}{D}}$ such that, for $\mathbf a\in O$, the maximal degree of the polynomials in a minimal
  grevlex Gr\"obner basis of $\varphi_{\mathbf a}(\mathcal
  I_r)$ is
$$D r(m-r) +(D-1)(n-r)(m-r)+1.$$
\end{cor}

\begin{thm}\label{theo:complposdim}
If $k > (n-r)(m-r)$, then there exists a non-empty Zariski open subset $O\subset \overline{\mathbb K}^{n m
  \binom{D-1+k}{D}}$
such that for any $\mathbf a\in O$, the arithmetic complexity of
computing a grevlex Gr\"obner basis of $\varphi_{\mathbf a}(\mathcal
I_r)$ is upper bounded by 
$$O\left(\binom{n}{r+1}\binom{m}{r+1}\binom{D r(m-r) +(D-1)(n-r)(m-r)+1+k}{k}^\omega\right).$$
\end{thm}

\begin{proof}
This is a consequence of Proposition \ref{prop:complF5} and Corollary \ref{coro:degposdim}.
\end{proof}

\subsection{The $0$-dimensional affine case}
\label{sec:affinecompl}
For practical applications, the affine case (i.e. when the entries of
the input matrix $\mathcal M$ are affine polynomials of degree $D$) is more often encountered
than the homogeneous one. In this case, the matrix $\mathcal M$ is defined as follows
$$\mathcal M=\begin{pmatrix}
  f_{1,1}&\dots&f_{1,m}\\ \vdots&\ddots&\vdots\\ f_{n,1}&\dots&f_{n,m}
\end{pmatrix}\hspace{2cm}f_{i,j} = \sum_{\ell=0}^D \sum_{t\in \mon(\ell,k)} \mathfrak a_{t}^{(i,j)} t.$$
We show in this section that the complexity results (Theorems
\ref{theo:complHom} and \ref{theo:complposdim}) still hold in the affine case. This is achieved by considering the homogenized system:

\begin{defn}\cite[Chapter 8, \S 2, Proposition 7]{CoxLitShe97}
  Let $(q_1,\ldots, q_\ell)\in \mathbb K[x_1,\ldots, x_k]^\ell$ be an affine
  polynomial system. We let $(\widetilde{q_1},\ldots, \widetilde{q_\ell})\in  \mathbb K[x_1,\ldots, x_k,x_{k+1}]^\ell$ denote
  its \emph{homogenized system} defined by
  $$\forall i, \text{ s.t. }1\leq i\leq \ell, \widetilde{q_i}(x_1,\ldots, x_k,x_{k+1})=x_{k+1}^{\deg(q_i)}q_i\left(\frac{x_1}{x_{k+1}},\ldots, \frac{x_k}{x_{k+1}}\right).$$
\end{defn}

Notice that if an affine polynomial system has solutions, then the dimension of the ideal generated by its homogenized system is positive.

The study of the homogenized system is motivated by the fact that, for
the grevlex ordering, the dehomogenization of a Gr\"obner basis of
$\langle\widetilde{q_1},\ldots, \widetilde{q_\ell}\rangle$ is a
Gr\"obner basis of $\langle q_1,\ldots, q_\ell\rangle$. Therefore, in
order to compute a Gr\"obner basis of the affine system, it is
sufficient to compute a Gr\"obner basis of the homogenized system (for which we have complexity estimates by Theorems \ref{theo:complHom} and \ref{theo:complposdim}).

To estimate the complexity of the change of ordering, we need bounds on the degree of the ideal in the affine case:

\begin{lem}\label{lem:degaff}
The degree of the ideal $\langle q_1,\ldots, q_\ell\rangle$ is upper bounded by that of $\langle\widetilde{q_1},\ldots, \widetilde{q_\ell}\rangle$.
\end{lem}

\begin{proof}
  The rings $\mathbb K[x_1,\ldots, x_k]/\langle q_1,\ldots,
  q_\ell\rangle$ and $\mathbb K[x_1,\ldots, x_k, x_{k+1}]/\langle
  \widetilde{q_1},\ldots, \widetilde{q_\ell}, x_{k+1}-1\rangle$ are
  isomorphic. Therefore the degrees of the ideals $\langle q_1,\ldots,
  q_\ell\rangle$ and $\langle \widetilde{q_1},\ldots,
  \widetilde{q_\ell}, x_{k+1}-1\rangle$ are equal. Since $\deg(x_{k+1}-1)=1$, we
  obtain:
$$\begin{array}{rcl}\DEG\left(\langle q_1,\ldots, q_\ell\rangle\right)&=&\DEG\left(\langle\widetilde{q_1},\ldots, \widetilde{q_\ell}, x_{k+1}-1\rangle\right)\\
&\leq&\DEG\left(\langle\widetilde{q_1},\ldots, \widetilde{q_\ell}\rangle\right).
\end{array}$$
\end{proof}

\begin{lem}\label{lem:dregaff}
The degree of regularity with respect to the grevlex ordering of the ideal $\langle q_1,\ldots, q_\ell\rangle$ is upper bounded by that of $\langle\widetilde{q_1},\ldots, \widetilde{q_\ell}\rangle$.
\end{lem}

\begin{proof}
Let $\chi$ denote the dehomogenization morphism:
$$\begin{array}{rrcl}
\chi:&\mathbb K[x_1,\ldots, x_{k+1}]&\longrightarrow&\mathbb K[x_1,\ldots, x_k]\\
&f(x_1,\ldots,x_k,x_{k+1})&\longmapsto&f(x_1,\ldots,x_k,1)
\end{array}$$
If $G$ is a grevlex Gr\"obner basis of $\langle\widetilde{q_1},\ldots, \widetilde{q_\ell}\rangle$, then $\chi(G)$ is a grevlex Gr\"obner basis of $\langle q_1,\ldots, q_\ell\rangle$ (this is a consequence of the following property of the grevlex ordering: $\forall f\in\mathbb K[x_1,\ldots, x_{k+1}]$ homogeneous, $\LM(\chi(f))=\chi(\LM(f))$). 
Also, notice that for each $g\in G$, any relation
$g=\sum_{i=1}^{\ell} q_i h_i$
gives a relation $\chi(g)=\sum_{i=1}^{\ell} \chi(q_i) \chi(h_i)$ of lower degree since 
$$\deg(\chi(q_i)\chi(h_i))\le \deg(q_i h_i).$$
Consequently, a Gr\"obner basis of $\langle q_1,\ldots, q_\ell\rangle$ can be obtained by computing the row echelon form of the Macaulay matrix of $(q_1,\ldots, q_\ell)$ in degree $\dreg(\langle\widetilde{q_1},\ldots, \widetilde{q_\ell}\rangle)$. Therefore, the degree of regularity with respect to the grevlex ordering of the ideal $\langle q_1,\ldots, q_\ell\rangle$ is upper bounded by that of $\langle\widetilde{q_1},\ldots, \widetilde{q_\ell}\rangle$.
\end{proof}

We can now state the main complexity result for the affine generalized MinRank problem:
\begin{thm}\label{theo:compl_affine}
Suppose that the matrix $\mathcal M$ contains generic affine polynomials of degree $D$:
$$\mathcal M=\begin{pmatrix}
  f_{1,1}&\dots&f_{1,m}\\ \vdots&\ddots&\vdots\\ f_{n,1}&\dots&f_{n,m}
\end{pmatrix}\hspace{2cm}f_{i,j} = \sum_{\ell=0}^D \sum_{t\in \mon(\ell,k)} \mathfrak a_{t}^{(i,j)} t.$$
There exists a non identically null polynomial $h\in \mathbb
K[\mathfrak a]$ such that for any $\mathbf a\in \overline{\mathbb K}^{n m
  \binom{D+k}{D}}$ such that $h(\mathbf a)\neq 0$, the overall
arithmetic complexity of computing the set of points such that the
matrix $\varphi_{\bf a}(\mathcal M)$ has rank less than $r+1$ with Gr\"obner
basis algorithms is upper bounded by
$$O\left(\binom{n}{r+1}\binom{m}{r+1}\binom{\dreg(\varphi_{\mathbf a}(\mathcal I_r))+k}{k}^\omega + k\left(\DEG(\varphi_{\mathbf a}(\mathcal I_r)\right)^3\right),$$
where $2\leq \omega\leq 3$ is a feasible exponent for the matrix
multiplication and
\begin{itemize}
\item if $k = (n-r)(m-r)$, then 
$$\dreg(\varphi_{\mathbf a}(\mathcal I_r))\leq D r(m-r) +(D-1)k+1,$$
$$\DEG(\varphi_{\mathbf a}(\mathcal I_r))\leq D^{(n-r)(m-r)} \prod_{i=0}^{m-r-1}\frac{i! (n+i)!}{(m-1-i)! (n-r+i)!}.$$
\item if $k < (n-r)(m-r)$, then assuming that Conjecture \ref{conj:froberg} is true, 
$$\dreg(\varphi_{\mathbf a}(\mathcal I_r))\leq \deg(P(t))+1,$$
 and $\DEG(\varphi_{\mathbf a}(\mathcal I_r))\leq P(1)$ where 
$$P(t)=\left[(1-t^D)^{(n-r)(m-r)} \frac{\det A(t^D)}{t^{D\binom{r}{2}}(1-t)^k}\right]_+.$$
\end{itemize}
\end{thm}

\begin{proof}
  This is a direct consequence of Proposition \ref{prop:complF5},
  Lemma \ref{lem:degaff}, Lemma \ref{lem:dregaff} and the complexity
  of the FGLM algorithm \cite{FauGiaLazMor93,FauMou11} ($O(k\DEG(\varphi_{\mathbf a}(\mathcal I_r)^3)$).
\end{proof}

\section{Case studies}
\label{sec:cases}
The aim of this section is to compare the complexity of the grevlex
Gr\"obner basis computation with the degree of the ideal in the
$0$-dimensional case (i.e. the number of solutions of the MinRank
problem counted with multiplicities). Since the ``arithmetic'' size
(i.e. the number of coefficients) of the lexicographical Gr\"obner
basis is close to the degree of the ideal in the $0$-dimensional case,
it is interesting to identify families of parameters for which the
arithmetic complexity of the computation is polynomial in this degree
under genericity assumptions.

Throughout this section, we focus on the $0$-dimensional case:
$k=(n-r)(m-r)$. Under genericity assumptions, we recall that, by
Corollary \ref{coro:degree} and Lemma \ref{lem:dregCompl},
$$\begin{array}{rcl}\dreg&=&Dr(m-r)+(D-1)k+1\\
  \DEG&=&\displaystyle D^{(n-r)(m-r)} \prod_{i=0}^{m-r-1}\frac{i! (n+i)!}{(m-1-i)! (n-r+i)!}.\end{array}$$
According to Theorem \ref{theo:compl_affine}, the complexity of the computation of the grevlex Gr\"obner basis is then upper bounded by 
$$O\left(\binom{n}{r+1}\binom{m}{r+1}\binom{Dr(m-r)+(D-1) k+1}{k}^\omega + k\left(\DEG\left(\varphi_{\mathbf a}\left(\mathcal I_r\right)\right)\right)^3\right).$$

In this section, $\Omega$ and $O$ are the Landau notations: for any positive
functions $f$ and $g$, we write $f=\Omega(g)$ (resp. $f=O(g)$) if there exists a
positive constant $C$ such that $f\geq C\cdot g$ (resp. $f\leq C\cdot g$).

\subsection{$D$ grows, $n$, $m$, $r$ are fixed}
\label{sec:Dgrows}
We first study the case where $n$, $m$ and $r$ are fixed (and thus $k=(n-r)(m-r)$ is constant too), and $D$ grows.
In that case, the arithmetic complexity of the grevlex Gr\"obner basis computation is
$O(D^{k\omega})$, and the degree is $\Omega(D^k)$. Therefore the arithmetic complexity has a polynomial dependence in the degree for these parameters.

\subsection{$n$ grows, $m, r, D$ are fixed}
\label{sec:casesmrDconst}
This paragraph is devoted to the study of the subfamilies of Generalized
MinRank problems when the parameters $m$, $r$ and $D$ are constant
values and $n$ grows.  Let $\ell$ denote the constant value
$\ell=m-r$. First, we assume that $D=1$. When $n$ grows, by Corollary
\ref{coro:degree} we have
$$\begin{array}{rcl}\log(\DEG)&=&
\displaystyle\log\left( \prod_{i=0}^{\ell-1}\frac{\binom{n+\ell-1}{r+i}}{\binom{n+\ell-1}{i}}\right)\\
&\underset{n\rightarrow\infty}{\sim}&r\ell\log(n)
\end{array}$$

On the other hand,
$$\begin{array}{rcl}
\log(\Compl)&=&\displaystyle\omega\log\binom{(n-r)\ell+ r\ell+1}{(n-r)\ell}+\log\binom{n}{r+1}+\log\binom{m}{r+1}\\
&=&\displaystyle\omega\log\binom{n\ell+1}{r\ell+1}+\log\binom{n}{r+1}+\log\binom{m}{r+1}\\
&\underset{n\rightarrow\infty}{\sim}&\left(\omega (r\ell+1)+r+1\right)\log(n).
\end{array}$$
Therefore, $\displaystyle\log(\Compl)/\log(\DEG)\underset{n\rightarrow\infty}{\sim}\frac{\omega (r\ell+1)+r+1}{r\ell}$ and hence the number of arithmetic operations is polynomial in the degree of the ideal.
 
Also, if $D\geq 2$ is constant, 
a similar analysis yields
$$\begin{array}{rcl}\log(\DEG)&=&
\displaystyle(n-r)\ell\log(D)+\log\left( \prod_{i=0}^{\ell-1}\frac{\binom{n+\ell-1}{r+i}}{\binom{n+\ell-1}{i}}\right)\\
&\underset{n\rightarrow\infty}{\sim}&\log(D)\ell n. \\
\log(\Compl)&=&\displaystyle\omega\log\binom{k+ D r\ell+(D-1)k+1}{k}+\log\binom{n}{r+1}+\log\binom{m}{r+1}\\
&=&\displaystyle\omega\log\binom{D n \ell+1}{(n-r)\ell}+\log\binom{n}{r+1}+\log\binom{m}{r+1}\\
&\underset{n\rightarrow\infty}{\sim}&\displaystyle\omega\log\binom{D n \ell}{n\ell}.
\end{array}$$
Then, using the fact that $\displaystyle\binom{\alpha n}{\beta n}\underset{n\rightarrow\infty}{\sim} n\left(\alpha\log(\alpha)-\beta\log(\beta)-(\alpha-\beta)\log(\alpha-\beta)\right)$, we obtain that
$$\log(\Compl)\underset{n\rightarrow\infty}{\sim} n\omega\ell(D\log(D)-(D-1)\log(D-1)).$$

Therefore, $\log(\Compl)/\log(\DEG)$ is upper bounded by a constant value and hence the arithmetic complexity of the Gr\"obner basis computation is also polynomial in the degree of the ideal for this subclass of Generalized MinRank problems under genericity assumptions.

\subsection{The case $r=m-1$}
\label{sec:reqm1}
The case $r=m-1$ is a special case of the setting studied in Section \ref{sec:casesmrDconst} which arises in several applications, since it is
the problem of finding at which points the evaluation of a polynomial
matrix is rank defective. In this setting, the formulas in Theorem
\ref{theo:compl_affine} are much simpler:

\begin{itemize}
\item the $0$-dimensional condition yields $k=n-m+1$;
\item $\dreg\leq Dn-(n-m)$;
\item $\DEG\leq \displaystyle D^{n-m+1}\binom{n}{m-1}$.
\end{itemize}

Therefore, the arithmetic complexity of the Gr\"obner basis computation is 
$$\Compl=O(\binom{n}{m}\binom{Dn+1}{n-m+1}^\omega).$$

If $D>1$ and $m$ are fixed,
$\log\left(\binom{n}{m}\binom{Dn+1}{n-m+1}^\omega\right)\underset{n\rightarrow\infty}{\sim}
m\log(n)+\omega\log\binom{Dn}{n}$ and a direct application of
Stirling's formula shows that
$$\omega\log\binom{Dn}{n}\underset{n\rightarrow\infty}{\sim} \omega (D\log D- (D-1)\log(D-1)) n.$$

On the other hand, $\log(\DEG)\underset{n\rightarrow\infty}{\sim}
n\log D$. Therefore, $\log(\Compl)/\log(\DEG)$ has a finite limit when
$n$ grows and $m$ is fixed, showing that, in this setting, the
arithmetic complexity is polynomial in the degree of the ideal.

\subsection{Experimental results}
\label{sec:expe}
In this section, we present some experimental results obtained by
using the Gr\"obner bases package FGb (using the $F_5$ algorithm) and
the implementation of the $F_4$ algorithm in the \textsc{Magma}
computer algebra system \cite{BosCanPla97}. All instances were
constructed as random (with uniform distribution) 0-dimensional
MinRank problems (i.e. $n m -(n+m-r)r=k$) over the finite field
$\mathbb F_{65521}$. All experiments were conducted on a 2.93 GHz Intel Xeon 
with 132 GB RAM.

\begin{table}
\begin{tabular}{|c||@{}c@{}|@{}c@{}|@{}c@{}|@{}c@{}|@{}c@{}|@{}c@{}|@{}c@{}|}
\hline
(n,m,D,r,k)&~$\mathsf{DEG}$~~&~$\dreg$~
&$F_4$ time(Magma)& FGLM time(Magma)& $F_5$ time/nb.ops(FGb)& FGLM time(FGb)\\
\hline
\hline
(6,5,2,4,2)&60&11&0.001s&0.001s&0.00s/$2^{13.32}$&0.00s\\
(6,5,3,4,2)&135&17&0.002s&0.019s&0.00s/$2^{15.29}$&0.00s\\
(6,5,4,4,2)&240&23&0.004s&0.09s&0.01s/$2^{16.79}$&0.01s\\
(5,5,2,3,4)&800&17&0.25s&6.3s&0.24s/$2^{25.56}$&0.19s\\
(8,5,2,4,4)&1120&13&0.7s&20s&0.43s/$2^{26.71}$&0.58s\\
(5,5,3,3,4)&4050&27&6.7s&567s&5.43s/$2^{30.68}$&3s\\
(6,5,2,3,6)&11200&19&479s&17703s&94.85s/$2^{35.7}$&203s\\
\hline
\end{tabular}
\caption{Experimental results \label{table:experiments}}
\end{table}

Useful information can be read from Table~\ref{table:experiments}.
First, the experimental values of the degree of regularity and of the
degree match exactly the theoretical values given in Lemma~\ref{lem:dregCompl} and in Corollary~\ref{coro:degree}.  Also, it can
be noted that the most relevant indicator of the complexity of the
Gr\"obner basis computation seems to be the degree of the ideal.

The comparison between the complexity bound and the degree of the
ideal is illustrated in Figures \ref{fig:D} and
\ref{fig:alphabeta}. First, Figure \ref{fig:D} shows that the bound on
the complexity of the Gr\"obner computation is polynomial in the
degree of the ideal when $D$ grows ($n=m=20$, $r=10$ fixed), since
$\log(\Compl_{\sf F_5})/\log(\DEG)$ is upper bounded by $5$. This is
in accordance with the analysis performed in Section \ref{sec:Dgrows}.

Then Figure \ref{fig:alphabeta} shows empirically that if
$m=\lfloor\beta n\rfloor$ and $r=\lfloor \alpha n\rfloor -1$ (with
$\alpha\leq\beta\leq 1$) and $n$ grows, then the complexity bound is
also polynomial in the degree of the ideal.

\begin{figure}
\includegraphics[width=0.6\linewidth]{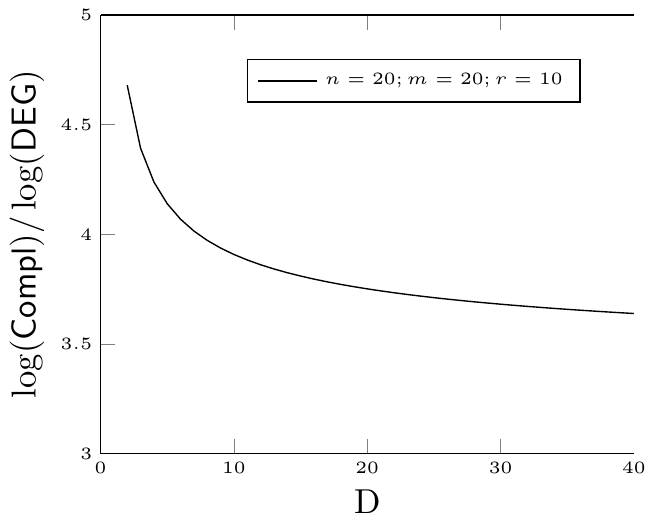}
\caption{Numerical values of $\log(\Compl_{\sf F_5})/\log(\DEG)$, for $n=m=20, r=10, k=(n-r)(m-r)$.}
\label{fig:D}
\end{figure}

\begin{figure}
\includegraphics[width=0.6\linewidth]{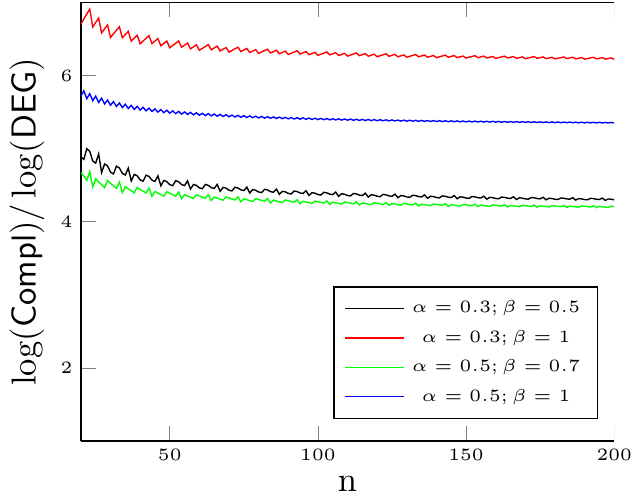}
\caption{Numerical values of $\log(\Compl_{\sf F_5})/\log(\DEG)$, for $m=\lfloor\beta n\rfloor, r=\lfloor\alpha n\rfloor-1, D=1, k=(n-r)(m-r)$.}
\label{fig:alphabeta}
\end{figure}

\medskip

However, there also exist families of generalized
MinRank problem where the complexity bound for the Gr\"obner basis
computation is \emph{not} polynomial in the degree of ideal. For
instance, taking $n=m$ and fixing the values of $r$ and $D$ yields such a family.

The experimental behavior of $\log(\Compl_{\sf F_5})/\log(\DEG)$ is
plotted in Figure \ref{fig:badcase}.  We would like to point out that
this does not necessarily mean that the complexity of the Gr\"obner
basis computation is not polynomial in the degree of the
ideal. Indeed, the complexity bound
$O\left(\binom{n}{r+1}\binom{m}{r+1}\binom{k+\dreg}{k}^\omega\right)$ is not sharp and the figure
only shows that the bound is not polynomial.

The problem of showing whether the actual arithmetic complexity of the $F_5$
algorithm is polynomial or not in the degree of the ideal for any
families of parameters of the generalized MinRank problem remains an
open problem.

\begin{figure}
\includegraphics[width=0.6\linewidth]{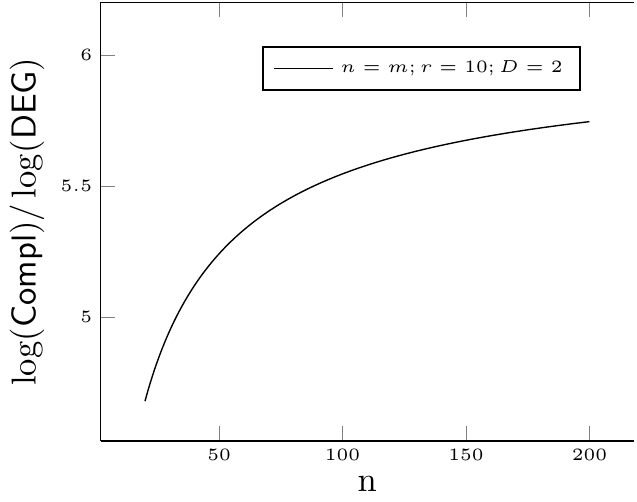}
\caption{Numerical values of $\log(\Compl_{\sf F_5})/\log(\DEG)$, for $m=\lfloor\beta n\rfloor, r=\lfloor\alpha n\rfloor-1, D=1, k=(n-r)(m-r)$.}
\label{fig:badcase}
\end{figure}

\section{Application to bi-homogeneous systems of bi-degree $(D,1)$}
\label{sec:bihomogeneous}

In this section, we show that the previous complexity analysis can be
used to obtain bounds on the complexity of solving bi-homogeneous
systems of bi-degree $(D,1)$ by using Gr\"obner bases algorithms.
These structured systems can appear naturally in some applications,
for instance in geometry and in optimization. Indeed the classical
technique of \emph{Lagrange multipliers} -- when used to optimize a
polynomial function under polynomial constraints -- gives rise to a
bi-homogeneous system of bi-degree $(D,1)$.

Bi-homogeneous polynomials are defined as follows: given two finite sets of variables $X=\{x_0,\ldots, x_{n_x}\}$ and $Y=\{y_0,\ldots, y_{n_y}\}$, a polynomial $f\in \mathbb K[X,Y]$ is called \emph{bi-homogeneous} if for any $\lambda,\mu\in \mathbb K$, there exist $d_x,d_y\in\mathbb N$ such that
$$f(\lambda X, \mu Y)=\lambda^{d_x}\mu^{d_y} f(X,Y).$$
The couple $(d_x, d_y)$ is called the \emph{bi-degree} of $f$.

In this section, we focus on generic systems of $n_x+n_y$
bi-homogeneous equations of bi-degree $(D,1)$. Such systems
have a finite number of solutions on the biprojective space $\mathbb
P^{n_x}\times \mathbb P^{n_y}$. One way to compute them is to start by computing their projection on $\mathbb P^{n_x}$, and then lift them to $\mathbb
P^{n_x}\times \mathbb P^{n_y}$ by solving linear systems (this can be done since the equations are linear with respect the variables $y_0,\ldots, y_{n_y}$).

The following proposition shows that computing the projection on $\mathbb P^{n_y}$ can be computed by solving a homogeneous MinRank problem.

\begin{prop}\label{prop:jacbihom}
Let $f_1,\ldots, f_{m}\in \mathbb K[X,Y]$ be a bi-homogeneous system of bi-degree $(D,1)$.
If $m>n_y$, then $(x_0:\ldots:x_{n_x}, y_0:\ldots:y_{n_y})\in \mathbb P^{n_x}\times\mathbb P^{n_y}$ is a zero of this system if and only if the matrix 
$$\jac_Y(x_0,\ldots, x_{n_x})=\begin{pmatrix}
\frac{\partial f_1}{\partial y_0}&\dots&\frac{\partial f_1}{\partial y_{n_y}}\\
\vdots&\vdots&\vdots\\
\frac{\partial f_m}{\partial y_0}&\dots&\frac{\partial f_m}{\partial y_{n_y}}
\end{pmatrix}$$
is rank defective.
\end{prop}

\begin{proof}
First, notice that
$$
\begin{pmatrix}
  f_1\\\vdots\\f_m
\end{pmatrix}=\jac_Y(x_0,\ldots, x_{n_x})\cdot
  \begin{pmatrix}y_0\\\vdots\\y_{n_y}
  \end{pmatrix}.$$ Therefore, $(x_0:\ldots:x_{n_x},
  y_0:\ldots:y_{n_y})\in \mathbb P^{n_x}\times\mathbb P^{n_y}$ is a zero of the system if and only if $(y_0,\ldots, y_{n_y})$ belongs to the
  kernel of $\jac_Y$. Since $m>n_y$, the number of rows is greater
  than or equal to the number of columns of $\jac_Y$, and hence
  $\jac_Y$ is rank defective.
\end{proof}

In applications, most of bi-homogeneous systems occurring are
\emph{affine}: A polynomial $f\in \mathbb K[x_1,\ldots,
x_{n_x},y_1,\ldots, y_{n_y}]$ is called affine of bi-degree $(D,1)$ if
there exists a bi-homoge\-neous polynomial $f^h\in \mathbb K[x_0,\ldots,
x_{n_x},y_0,\ldots, y_{n_y}]$ of bi-degree $(D,1)$ such that
$$f(x_1,\ldots, x_{n_x}, y_1,\ldots, y_{n_y})=f^h(1,x_1,\ldots, x_{n_x},1,y_1,\ldots, y_{n_y}).$$
This means that each monomial occurring in $f$ has bi-degree $(i,j)$
with $i\leq D$ and $j\leq 1$. Notice that the polynomial $f^h$ is
uniquely defined and that Proposition \ref{prop:jacbihom} also holds
in the affine context:

\begin{prop}\label{prop:jacbiaff}
  Let $f_1,\ldots, f_{m}\in \mathbb K[x_1,\ldots, x_{n_x},y_1,\ldots,
  y_{n_y}]$ be an affine system of bi-degree $(D,1)$.  If $m>n_y$ and
  $(x_1,\ldots,x_{n_x}, y_1,\ldots,y_{n_y})\in \mathbb
  K^{n_x}\times\mathbb K^{n_y}$ is a zero of the system, then the
  $m\times (n_y+1)$ matrix
$$\jac^a_Y(x_1,\ldots, x_{n_x})=\begin{pmatrix}
  f_1(x_1,\ldots, x_{n_x},0,\ldots, 0)&\frac{\partial f_1}{\partial y_1}&\dots&\frac{\partial f_1}{\partial y_{n_y}}\\
  \vdots&\vdots&\vdots\\
  f_m(x_1,\ldots, x_{n_x},0,\ldots, 0)&\frac{\partial f_m}{\partial
    y_0}&\dots&\frac{\partial f_m}{\partial y_{n_y}}
\end{pmatrix}$$
is rank defective.
\end{prop}

\begin{proof}
The proof is similar to that of \ref{prop:jacbihom} since
$$
\begin{pmatrix}
  f_1\\\vdots\\f_m
\end{pmatrix}=\jac_Y^a(x_1,\ldots, x_{n_x})\cdot
  \begin{pmatrix}1\\y_1\\\vdots\\y_{n_y}
  \end{pmatrix}.$$ Therefore, if $(x_1,\ldots,x_{n_x},
  y_1,\ldots,y_{n_y})$ is a zero of the system then there is a non-zero vector in the
  kernel of $\jac_Y^a$ (however in the affine case, the converse is not true).
\end{proof}

An algebraic description of the variety $V$ of a $0$-dimensional polynomial system can be obtained by computing a rational parametrization, i.e. a polynomial $g(u)\in \mathbb K[u]$ and a set of rational functions $g_1,\ldots, g_{n_x},h_1,\ldots, h_{n_y}\in \mathbb K(u)$ such that
$$\begin{array}{c}(x_1,\ldots, x_{n_x},y_1,\ldots, y_{n_y})\in V\\ \Updownarrow\\ \exists u\in \mathbb K, s.t. g(u)=0, \forall i\in \{1,\ldots, n_x\}, x_i=g_i(u), \forall j\in \{1,\ldots, n_y\}, y_j=h_j(u).\end{array}$$

To obtain a rational parametrization, we need a separating element: a linear form which takes different values on all points of $V$. Therefore, a rational parametrization exists only if the cardinality of the field $\mathbb K$ is infinite or large enough. 

\begin{algorithm}
\caption{Rational parametrization of systems of bi-degree $(D,1)$}
\begin{algorithmic}[1]
  \Require $f_1,\ldots, f_{n_x+n_y}\in \mathbb K[X,Y]$ a system of
  affine polynomials of bi-degree $(D,1)$ such that the ideal they
  generate is radical and $0$-dimensional; \newline $(\alpha_1,\ldots, \alpha_{n_x-1})\in\mathbb K^{n_x-1}$; \newline a full rank matrix $M=(m_{i,j})\in \mathbb K^{ny\times(n_x+n_y)}$.  \Ensure Returns a rational
  parametrization of the variety of the system or ``fail''.  \State Compute for
  each $i\in \{1,\ldots, n_x+n_y\}$, $$\widetilde{f_i}(x_1,\ldots,
  x_{n_x-1},u,y_1,\ldots, y_{n_y})=f_i(x_1,\ldots,
  x_{n_x-1},u-\displaystyle\sum_{\ell=1}^{{n_x}-1}\alpha_\ell
  x_\ell,y_1,\ldots, y_{n_y}).$$ \State Compute the matrix
  $\jac^a_Y(\widetilde{f_1},\ldots, \widetilde{f_{n_x+n_y}})$.  \State
  Compute a lex Gr\"obner basis $G$ of the ideal $I\subset\mathbb
  K[x_1,\ldots, x_{n_x-1},u]$ generated by the maximal minors of the
  matrix $\jac^a_Y(\widetilde{f_1},\ldots,
  \widetilde{f_{n_x+n_y}})$. If the Gr\"obner basis has the following
  shape (the \emph{shape position}):
$$\begin{array}{r}
x_1-g_1(u)\\
x_2-g_2(u)\\
\vdots\\
x_{n_x-1}-g_{n_x-1}(u)\\
g(u),
\end{array}$$
then continue to Step 4, else return ``fail''.

\State Using $M$, compute a linear combination of the polynomials of the system evaluated at $(g_1(u),\ldots, g_{n_x-1}(u))$:
$$\begin{pmatrix}
\widehat{f_1}(y_1,\ldots, y_{n_y},u)\\
\vdots\\
\widehat{f_{n_y}}(y_1,\ldots, y_{n_y},u)
\end{pmatrix} = M\cdot \begin{pmatrix}
\widetilde{f_1}(g_1(u),\ldots, g_{n_x-1}(u), u,y_1,\ldots, y_{n_y})\mod g(u)\\
\vdots\\
\widetilde{f_{n_x+n_y}}(g_1(u),\ldots, g_{n_x-1}(u),u,y_1,\ldots, y_{n_y})\mod g(u)
\end{pmatrix}$$

\State If the linear system $\widehat{f_1}=\ldots=\widehat{f_{n_y}}=0$ has rank $n_y$ (as a linear system in
$\mathbb K(u)[Y]$ where the variables are $y_1,\ldots,y_{n_y}$), continue to Step 6, else return ``fail''.

\State Using Cramer's rule, solve the system
$\widehat{f_1}=\ldots=\widehat{f_{n_y}}=0$ as a linear system in
$\mathbb K(u)[Y]$. This yields rational functions $h_i(u)\in \mathbb
K(u)$ such that, for $i\in\{1,\ldots, n_y\}$, $y_i-h_i(u)=0$.  
\State
Return the rational parametrization
$$\begin{array}{cc}
g(u)=0&\\
x_1=g_1(u)&y_1=h_1(u)\\
\vdots&\vdots\\
x_{n_x-1}=g_{n_x-1}(u)&y_{n_y-1}=h_{n_y-1}(u)\\
x_{n_x}=u-\displaystyle\sum_{\ell=1}^{{n_x}-1}\alpha_\ell g_\ell(u)&y_{n_y}=h_{n_y}(u)
\end{array}$$
\end{algorithmic}\label{algo:bihom}
\end{algorithm}

Under the assumption that the field $\mathbb K$ is sufficiently large,
Algorithm \ref{algo:bihom} uses the property described in Proposition
\ref{prop:jacbiaff} to find a rational parametri\-zation of the zeroes
of a radical and $0$-dimensional system of $n_x+n_y$ affine
polynomials of bi-degree $(D,1)$. The algorithm proceeds by computing first a rational parametrization of the projection of the zero set on $\mathbb K^{n_x}$. This is done by computing a lexicographical Gr\"obner basis of a Generalized MinRank Problem. Then this parametrization is lifted to the whole space by solving a linear system (this can be done since the equations are linear with respect to the variables $y_1,\ldots, y_{n_y}$).

The success of Algorithm \ref{algo:bihom} depends on the choice of the
parameters $\alpha$ (a linear change of coordinates such that $x_n$ is
a separating element) and $M$.  However, as we will see in Theorem
\ref{theo:complbihom}, if the cardinality of $\mathbb K$ is infinite
or large enough, then almost all choices of $\alpha$ and $M$ are
good. Therefore, these parameters can be chosen at random.  If
Algorithm \ref{algo:bihom} unluckily fails, then it can be restarted
with the same algebraic system and different values of $\alpha$ and
$M$.

We now prove that the complexity of Algorithm \ref{algo:bihom} is
bounded by the complexity of the underlying generalized MinRank
problem and that most choices of $(\alpha_1,\ldots, \alpha_{n_x-1})$
and $M$ do not fail. 

\begin{thm}\label{theo:complbihom}
  Let $f_1,\ldots, f_{n_x+n_y}\in \mathbb K[X,Y]$ be an affine system
  of bi-degree $(D,1)$ such that the ideal $\langle f_1,\ldots,
  f_{n_x+n_y}\rangle$ is radical and $0$-dimensional. Then there
  exists non-identically null polynomials $h_1\in\mathbb K[z_1,\ldots,
  z_{n_x-1}]$ and $h_2\in\mathbb K[z_{1,1},\ldots, z_{n_y,n_x+n_y}]$ such
  that, for any choice of $(\alpha_1,\ldots, \alpha_{n_x-1})$ and
  $M=(m_{i,j})\in\mathbb K^{n_y\times (n_x+n_y)}$ verifying:
\begin{itemize}
\item the matrix $\jac^a_Y(\widetilde{f_1},\ldots, \widetilde{f_{n_x+n_y}})$ verifies the conditions of Theorem \ref{theo:compl_affine};
\item $h_1(\alpha_1,\ldots, \alpha_{n_x-1})h_2(m_{1,1},\ldots, m_{n_y,n_x+n_y})\neq 0$,
\end{itemize}
Algorithm \ref{algo:bihom} returns a rational parametrization of the variety of the system and its complexity is upper bounded by 
$$O\left(\binom{n_x+n_y}{n_x-1}\binom{D(n_x+n_y)+1}{n_x}^\omega + n_x\left(D^{n_x}\binom{n_x+n_y}{n_x}\right)^3\right).$$
\end{thm}

\begin{proof}
In this proof, $\widetilde O()$ stands for the \emph{soft-Oh} notation: if $f$ and $g$ are positive functions, $f=\widetilde O(g)$ means that there exists $k\in\mathbb N$ such that $f=O(g\cdot \log^k(g))$.
  Let $I$ denote the ideal generated by $f_1,\ldots, f_{n_x+n_y}$.
  According to \cite{Lak90,BecMorMarTra94}, for any radical
  $0$-dimensional ideal, there exists a polynomial $h_1$ such that if
  $h_1(\alpha_1,\ldots, \alpha_{n_x-1})\neq 0$, then the system is in shape position after
  the change of coordinates 
$$x_{n_x}\mapsto x_{n_x}-\sum_{\ell=1}^{n_x-1} \alpha_\ell x_\ell.$$

The polynomial $h_2$ is
  chosen such that if $h_2(m_{i,j})\neq 0$, then the linear system
  $\widehat{f_1}=\dots=\widehat{f_{n_y}}=0$ in $\mathbb K(u)[Y]$ has
  rank exactly $n_y$. Consider now the following linear system (where the variables are $y_1,\ldots, y_{n_y}$):
  $$\begin{pmatrix}z_{1,1}&\dots&z_{1,n_x+n_y}\\\vdots&\vdots&\vdots\\
z_{n_y,1}&\dots&z_{n_y,n_x+n_y}\end{pmatrix}\cdot \begin{pmatrix}
\widetilde{f_1}(g_1(u),\ldots, g_{n_x-1}(u), u,y_1,\ldots, y_{n_y})\mod g(u)\\
\vdots\\
\widetilde{f_{n_x+n_y}}(g_1(u),\ldots, g_{n_x-1}(u),u,y_1,\ldots, y_{n_y})\mod g(u)
\end{pmatrix}=0.$$ Its determinant (which lies in $\mathbb
K[z_{1,1},\ldots, z_{n_y,n_x+n_y},u]$) is not zero since the ideal
generated by the input system $(f_1,\ldots, f_{n_x+n_y})$ is
$0$-dimensional and proper. By considering this determinant as a polynomial in $\mathbb K[z_{1,1},\ldots, z_{n_y,n_x+n_y}][u]$, the polynomial $h_2\in \mathbb K[z_{1,1},\ldots, z_{n_y,n_x+n_y}]$ is chosen as a non-zero coefficient of a term $u^\beta$. Consequently, the algorithm does not fail if $h_1(\alpha_1,\ldots, \alpha_{n_x-1})\neq 0$ and $h_2(m_{i,j})\neq 0$.

Now we proceed with the complexity analysis:
\begin{itemize}
\item the complexity of the substitution step to compute the polynomials $\widetilde{f_i}$ is upper bounded by $\widetilde O((n_x+n_y)D n_x n_y)$.
\item By Theorem \ref{theo:compl_affine}, the complexity of the Gr\"obner basis computation is upper bounded by $$O\left(\binom{n_x+n_y}{n_x-1}\binom{D(n_x+n_y)+1}{n_x}^\omega + n_x\left(\DEG(I)\right)^3\right).$$
\item Since $\deg(g_{n_x})\leq \DEG(I)$, a monomial $u^{n_x}\prod_{i=1}^{n_x-1} x_i^{\alpha_i}$ of degree $D$ can be
  evaluated in the univariate polynomials $(g_1(u),\dots, g_{n_x-1}(u))$ modulo $g(u)$
  in complexity $\widetilde O(D\DEG(I))$ by using a subproduct tree
  \cite{BosSch05}, quasi-linear multiplication of univariate
  polynomials and quasi-linear modular reduction. Since there are at most
  $(n_x+n_y)(n_y+1)\binom{n_x+D}{n_x}$ such monomials in the system
  $f_1,\ldots, f_{n_x+n_y}$, the Step 4 of the Algorithm needs at most
  $$\widetilde O\left((n_x+n_y)n_y\binom{n_x+D}{n_x}D \DEG(I)\right)$$
  arithmetic operations in $\mathbb K$. \\Notice that
  $n_x+n_y\leq \binom{n_x+n_y}{n_x-1}$ and $\DEG(I)\leq
  \binom{D(n_x+n_y)+1}{n_x}$.
\begin{itemize}
\item If $D\geq 2$:
  for any $a, b, c\in \mathbb N$ such
  that $b<a$, $\binom{a}{b}c\leq \binom{a+c}{b}$. Therefore, $D
  n_y\binom{n_x+D}{n_x}\leq\binom{n_x+n_y+ 2 D}{n_x}$. Also, notice
  that, for $D\geq 2$ and for any $n_x, n_y$ such that $n_x n_y>1$,
  $n_x+n_y+2 D\leq D(n_x+n_y)+1$.
  Therefore,
$$\widetilde O\left((n_x+n_y)n_y\binom{n_x+D}{n_x}D \DEG(I)\right)\leq \widetilde O\left(\binom{n_x+n_y}{n_x-1}\binom{D(n_x+n_y)+1}{n_x}^2\right).$$
\item If $D=1$: $(n_x+n_y)n_y\binom{n_x+1}{n_x}=(n_x+n_y)n_y n_x$ is bounded by $\binom{n_x+n_y}{n_x-1}\binom{(n_x+n_y)+1}{n_x}$.
\end{itemize}
Therefore, the complexity of the Step 4 of Algorithm \ref{algo:bihom} is upper bounded by the complexity of the Gr\"obner basis computation: $O\left(\binom{n_x+n_y}{n_x-1}\binom{D(n_x+n_y)+1}{n_x}^\omega\right).$

\item To solve the linear system by using Cramer's rule, we need to
  compute $n_x+1$ determinants of $(n_x\times n_x)$-matrices whose
  entries are univariate polynomials of degree $D$. This can be
  achieved by using a fast evaluation-interpolation strategy with
  complexity $\widetilde O\left(D n_x^{\omega +1}\right)$ (since multi-set evaluation and interpolation of univariate polynomials can be done in quasi-linear time, see e.g. \cite{BosSch05}).
\end{itemize}
Since $\DEG(I)$ is bounded by $D^{n_x}\binom{n_x+n_y}{n_x}$, the sum of all these complexities is upper bounded by 
$$O\left(\binom{n_x+n_y}{n_x-1}\binom{D(n_x+n_y)+1}{n_x}^\omega + n_x\left(D^{n_x}\binom{n_x+n_y}{n_x}\right)^3\right).$$
\end{proof}

\begin{rem} According to \cite[Lemma 15]{FauSafSpa11} and \cite[Lemma
  16]{FauSafSpa11}, if $D=1$, there exists a non-empty Zariski open
  subset $O_1$ of the set of systems of bi-degree $(1,1)$, such that any system $(f_1,\ldots, f_{n_x+n_y})\in O_1$ is $0$-dimensional and
  radical. This statement also holds for systems of bi-degree $(D,1)$ with $D\in \mathbb N$, and the proof is similar.
\end{rem}

\subsection*{Acknowledgments}
This work was supported in part by the HPAC grant and the GeoLMI grant
(ANR 2011 BS03 011 06) of the French National Research Agency.  The
second author is member of the Institut Universitaire de France. We
wish to thank anonymous referees for their comments and suggestions.

\bibliographystyle{abbrv}
\bibliography{biblio}
\end{document}